\documentclass[12pt,oneside,english]{amsart}
\usepackage[T1]{fontenc}
\usepackage[utf8]{inputenc}
\usepackage[letterpaper]{geometry}
\usepackage{float}
\usepackage{booktabs}
\usepackage{amstext}
\usepackage{amsthm}
\usepackage{amssymb}
\usepackage{graphicx}
\usepackage{setspace}
\usepackage{xr}

\usepackage{verbatim}
\usepackage[authoryear]{natbib}
\makeatletter

\usepackage{latexsym,amsfonts,bm,epsfig,amsmath,thmtools}
\usepackage{dsfont}
\usepackage{graphicx}
\usepackage{color}
\usepackage[colorlinks,linkcolor=blue,citecolor=blue,bookmarks=false,pagebackref]{hyperref}
\usepackage{float}
\usepackage[margin=1cm]{caption}
\usepackage{xr}
\usepackage{algorithm}
\usepackage{algorithmic}
\usepackage{setspace}
\usepackage [displaymath,mathlines]{lineno}

\newcommand{\argmax}{\operatornamewithlimits{argmax\,}}

\newcommand{\E}{\mathbb{E}}

\newcommand{\Var}{\mathrm{Var}}

\mathchardef\mhyphen="2D

\def \x{x_{1:n}}
\def \y{y_{1:n}}
\def \a{\alpha_{1:n}}
\def \t{\theta}

\def \E{\mathbb{E}}

\def \dt {\mathrm{d}}

\def\tr{\text{\rm Tr}}

\newtheorem{assumption}{Assumption}
\newtheorem{remark}{Remark}
\newtheorem{theorem}{Theorem}

\newtheorem{lemma}{Lemma}
\newtheorem{example}{Example}

\makeatother

\begin{document}
	\title[Covering all your Bayes(es)]{Reliable Bayesian Inference in Misspecified Models}
	\author{David T. Frazier, Robert Kohn, Christopher Drovandi and David Gunawan}
	\thanks{Frazier: Department of Econometrics and Business Statistics, Monash University, Clayton VIC 3800, Australia. Kohn: Australian School of Business, School of Economics, University of New South Wales, Sydney NSW 2052, Australia. Drovandi: School of Mathematical Sciences, Centre for Data Science, Queensland University of Technology, Brisbane QLD 4000, Australia. Gunawan: : School of Mathematics and Applied Statistics, University of Wollongong. Corresponding author:  david.frazier@monash.edu}
	%\affil[1]{Department of Econometrics and Business Statistics, Monash University, Clayton VIC 3800, Australia}
%\author[2]{Robert Kohn}
%\affil[1]{Department of Econometrics and Business Statistics, Monash University, Clayton VIC 3800, Australia}
%\affil[2]{Australian School of Business, School of Economics, University of New South Wales, Sydney NSW 2052, Australia}

\begin{abstract}
We provide a general solution to a fundamental open problem in Bayesian inference, namely poor uncertainty quantification, from a frequency standpoint, of Bayesian methods in misspecified models. While existing solutions are based on explicit Gaussian approximations of the posterior, or computationally onerous post-processing procedures, we demonstrate that correct uncertainty quantification can be achieved by replacing the usual posterior with an intuitive approximate posterior. Critically, our solution is applicable to likelihood-based, and generalized, posteriors as well as cases where the likelihood is intractable and must be estimated. We formally demonstrate the reliable uncertainty quantification of our proposed approach, and show that valid uncertainty quantification is not an asymptotic result but occurs even in small samples. We illustrate this approach through a range of examples, including linear, and generalized, mixed effects models.

\end{abstract}

\maketitle

%	\spacingset{1.8} % DON'T change the spacing!

\section{Introduction}

Bayesian methods are lauded for their ability to tackle complicated models, deftly handle latent variables, and for providing a holistic (if subjective) treatment of the uncertainty regarding all model unknowns. While subjective in nature, under well-known regularity conditions, the expressions of uncertainty obtained using Bayesian methods asymptotically agree with frequentist methods. However, when the model underlying the (Bayesian) posterior update is misspecified the agreement between the Bayesian and frequentist statistical paradigms is lost, and expressions of uncertainty obtained using Bayesian methods are no longer `valid' from a frequentist standpoint. 

To correct this issue, several posterior post-processing methods have been suggested (see, .e.g., \citealp{muller2013risk}, \citealp{syring2019calibrating}, and \citealp{matsubara2021robust}). In general, all such procedures of which we are aware amount to the application of ex-post, frequentist, principles to the output of a Bayesian learning algorithm. Herein, we propose an intrinsically Bayesian solution to the fundamental problem of producing Bayesian inferences that are well-calibrated from a frequency standpoint.

Inspired by the literature on generalized Bayesian inference (\citealp{bissiri2016general}; \citealp{chernozhukov2003mcmc}), we devise a new type of posterior approximation that, when used in likelihood-based settings delivers equivalent behavior to `exact' Bayesian methods when the model is correctly specified, but produces posteriors whose credible sets are `valid' - from a frequentist standpoint - regardless of model specification. This approach is not based on any ad-hoc correction of posterior draws, or other extrinsic calibration techniques, and can be equally applied when the likelihood must be estimated. We demonstrate theoretically that under weak regularity conditions, and even when the likelihood is estimated, our proposed approach delivers valid frequentist uncertainty quantification regardless of model specification. Hence, we give for the first time a general framework for conducting Bayesian inference in possibly misspecified models that also delivers valid (frequentist) uncertainty quantification.  

However, unlike frequentist methods, our Bayesian approach delivers valid uncertainty quantification without needing to calculate second-derivative information. Thus, in cases where second derivatives are difficult to obtain, as in models with many latent variables, this approach provides a useful alternative to frequentist methods that must explicitly calculate these derivatives to correctly quantify uncertainty.

When the model is misspecified, Bayesian methods may not deliver inferences that are `fit for purpose'. To circumvent this issue, several researchers have  suggested using generalized Bayesian procedures that produce posteriors  based on loss functions that are specific to the task at hand; see, e.g., \citet{syring2020robust}, \cite{matsubara2021robust}, \cite{jewson2021general}, and \cite{loaiza2021focused} for examples. Unfortunately, as discussed by several authors, see, e.g., \cite{miller2021asymptotic} and \cite{syring2019calibrating}, generalized Bayesian posteriors are not well-calibrated: a credible set for a quantity of interest with posterior probability $(1-\alpha)$ contains the true quantity with actual probability - calculated under the true data generating process - smaller or larger than $(1-\alpha)$. Several approaches have been suggested to solve this issue; see, e.g., \cite{holmes2017assigning}, \cite{syring2019calibrating}, and see \cite{wu2020comparison} for a review of these methods. However, the proposed methods of which we are aware amount to the application of an extrinsic, and ad-hoc, correction to the output of the Bayesian learning algorithm.

We demonstrate that our approach can be directly adapted to the context of generalized Bayesian inference, where the likelihood function is replaced by a generic loss function, or a quasi-likelihood that is, at best, an approximation to the likelihood. The result is a class of posteriors for generalized Bayesian inference that produce loss-based credible sets that are well-calibrated. As such, we give, for the first time, a generalized Bayesian posterior whose credible sets are guaranteed to have the correct width, and which {does not rely on the use of expensive post-processing techniques.} 

The remainder of the paper proceeds as follows. Section \ref{sec:misspec} discusses the general issue of model misspecification in likelihood-based Bayesian inference, and demonstrates how a particular generalized posterior approach overcomes the known issues with Bayesian inference in this setting. Section \ref{sec:ext} extends this posterior approach to deal with situations where the likelihood depends on unobservable latent variables, and we show that even when the likelihood must be estimated our proposed approach delivers well-calibrated inferences. Section \ref{sec:general} demonstrates that the generalized posterior construction used in the likelihood-based Bayesian setting also applies in the context of posteriors built using general loss functions (as in \citealp{bissiri2016general}). Section \ref{sec:discuss} concludes the paper. Proofs of all stated results are given in the supplementary material. 

\section{Bayesian Inference in Misspecified Models}\label{sec:misspec}
%\subsection{Setup and Known Issues}
The observed data $\y=(y_1,\dots,y_n)^\top$, where $y_i\in\mathcal{Y}\subseteq\mathbb{R}^{d_y}$ ($i=1,\dots,n$), is generated from some true unknown distribution $P^{(n)}_0$. Since $P^{(n)}_0$ is unknown, we approximate the distribution of $\y$ using a class  of models $\{P^{(n)}_\theta:\theta\in\Theta\}$, which depends on unknown parameters $\theta\in\Theta\subseteq\mathbb{R}^{d_\theta}$. We assume there exists a measure that dominates both joint distributions $P^{(n)}_0$ and $P^{(n)}_\theta$ so that the joint densities $p^{(n)}_0$ and $p_\theta^{(n)}$ exist for all $n\ge1$. Our prior beliefs on $\theta$ are expressed via the probability density function $\pi(\theta)$. The prior beliefs are updated upon the observation of $\y$ via Bayes rule, to produce the exact posterior
\begin{equation}\label{eq:post}
	\pi(\theta\mid\y)\propto \pi(\theta)\exp\{-\ell_{n}(\theta)\},\text{ where } \ell_n(\theta):=-\log p^{(n)}_\theta(\y).
\end{equation}

When $p^{(n)}_\theta(\y)$ can be analytically evaluated there exist many ways to obtain samples from $\pi(\theta\mid\y)$. However, in many interesting situations within Bayesian inference $p^{(n)}_\theta(\y)$ is analytically intractable. This often occurs in models that depend on unobservable, or latent, variables  $\a =(\alpha_1^\top,\dots,\alpha_n^\top)^\top$, where $\a\sim p(\a\mid\theta)$, and $\alpha_i\in\mathcal{A}\subseteq\mathbb{R}^{d_\alpha}$ ($i=1,\dots,n$). The observed-data likelihood $p^{(n)}_\theta(\y)$ is then obtained by  integration over $\a$ in the complete-data likelihood $p_\theta(\y,\a)$: $$
p^{(n)}_\theta(\y)=\int_{\mathcal{A}}p_\theta(\y,\a)\dt\a.
$$  If the integration cannot be performed analytically it often remains feasible to obtain an estimator $\widehat{p}_\theta(\y\mid z)$ that depends on simulated random variables $z=(z_1,\dots,z_N)^\top\in\mathcal{Z}^{ }$, where $N$ indexes the number of `simulated draws' used to construct $\widehat{p}_\theta(\y\mid z)$. The use of $\widehat{p}_\theta(\y\mid z)$ within an MCMC algorithm results in a pseudo-marginal algorithm (\citealp{andrieu2009pseudo}) that targets a joint posterior over $(\theta^\top,z^\top)^\top$: 
\begin{equation*}
		\widehat\pi_{N}(\theta, z\mid \y) =\widehat{p}_\theta(\y\mid z)\pi(\theta) h(z\mid\theta),
\end{equation*}where $h(z\mid\theta)$ denotes the conditional density for the simulated data $z$. If $\int \widehat{p}_\theta(\y\mid z)h(z\mid\t)\dt z=p_\theta^{(n)}(\y)$, the marginal posterior $\widehat\pi_N(\theta\mid\y)=\int_{\mathcal{Z}} 	\widehat\pi_{N}(\theta, z\mid \y)\dt z$ agrees with the exact posterior $\pi(\theta\mid\y)$ (\citealp{andrieu2009pseudo}).  

Regardless of whether $p^{(n)}_\t$ must be estimated, when $P_0^{(n)}\notin \{P_\t^{(n)}:\t\in\Theta\}$ posterior inference is not generally `well-calibrated'. To state precisely what we mean by well-calibrated inferences, first define $\theta_\star:=\argmax_{\theta\in\Theta}\lim_n\E_{}\log \{p^{(n)}_\theta/p^{(n)}_0\}$ as the value of $\theta\in\Theta$ that minimizes the (limiting) Kullback-Liebler divergence from $P^{(n)}_0$ to $\{P^{(n)}_\theta:\theta\in\Theta\}$, and where $\E(\cdot)$ denotes the expectation operator under $P^{(n)}_0$. A credible set for $\theta_\star$ based on $\pi(\theta\mid\y)$ and having posterior probability $(1-\alpha)$  is said to be \textit{well-calibrated} if the set asymptotically contains $\theta_\star$ with $P^{(n)}_0$-probability $(1-\alpha)$. Such a definition directly extends to any function of $\theta_\star$. 

To illustrate why $\pi(\t\mid\y)$ is not well-calibrated in general, we require a few additional definitions. For a twice differentiable function $f:\Theta\rightarrow\mathbb{R}$, let $\nabla_\theta f(\theta)$ denote the gradient of $f(\theta)$, and $\nabla_{\theta\theta}^2f(\theta)$ its Hessian matrix. For $\ell_n(\theta)$ as defined in \eqref{eq:post}, let $m_n(\theta):=-\nabla_\theta\ell_n(\theta)$ and $\mathcal{H}_n(\theta):=n^{-1}\nabla_{\theta\theta}^2\ell_n(\theta)$ denote the gradient and (minus the) Hessian, respectively; while $\mathcal{H}(\theta):=\lim_n\E[\mathcal{H}_n(\theta_\star)]$ denotes the expected observed information, and $\mathcal{I}(\theta):=\lim_n\text{Cov}\{m_n(\theta)/\sqrt{n}\}$ denotes the Fisher information. 

For $\hat\theta_n$ the value of $\theta\in\Theta$ that solves $0=m_n(\theta)$, i.e., the maximum likelihood estimator, under classical regularity conditions, see \cite{white1982maximum} or \cite{kleijn2012bernstein}, it is known that
$
\sqrt{n}(\hat\theta_n-\theta_\star)\Rightarrow N(0,\Sigma^{-1}_\star)$, where $\Sigma^{-1}_\star:=\mathcal{H}(\theta_\star)^{-1} \mathcal{I}(\theta_\star)\mathcal{H}(\theta_\star)^{-1}
$,  $N(\mu,\Sigma)$ denotes the Gaussian cumulative distribution function (CDF) with mean $\mu$ and variance $\Sigma$, and $\Rightarrow$ denotes weak convergence under $P^{(n)}_0$. Furthermore, the posterior satisfies the Bernstein-von Mises result
\begin{flalign*}
P^{(n)}_0\sup_{B\in\Theta}\left| \int_{B}\left[\pi(\vartheta\mid \y)-N\{\vartheta;\hat\theta_n,(n\mathcal{H}_\star)^{-1}\}\right]\dt\vartheta\right|\rightarrow{}0,
\end{flalign*}where $N\{x;\mu,\Sigma\}$ denotes the Gaussian probability density function (PDF) with mean $\mu$ and variance $\Sigma$; 
see, e.g., Theorem 2.1 in \cite{kleijn2012bernstein}. The Bernstein-von Mises result demonstrates that the `width' of credible sets based on $\pi(\t\mid\y)$ are determined by $\mathcal{H}(\theta_\star)^{-1}$, while the first result states that (asymptotically) valid frequentist confidence sets have `width' determined by the sandwich covariance $\Sigma_\star^{-1}$. Thus, credible sets for $\theta_\star$ with posterior probability $(1-\alpha)$ will not contain $\theta_\star$ with $P^{(n)}_0$-probability $(1-\alpha)$ in general,  unless $\mathcal{H}(\t_\star)=\mathcal{I}(\t_\star)$. Consequently, posterior credible sets are not well-calibrated, and Bayesian uncertainty quantification may be too optimistic to be practically useful.  

The lack of well-calibrated credible sets in misspecified models has led researchers to consider many different approaches to `correct' this issue. However, each of the suggested approaches of which we are aware are either based on computationally onerous bootstrapping procedures, or amount to an explicit Gaussianity assumption on the posterior and the application of ex-post corrections of $\pi(\theta\mid\y)$ to correct the coverage. Section \ref{sec:discuss1} discusses several such methods.

\subsection{Reliable Bayesian Uncertainty Quantification}\label{sec:cov}

If we are willing to move away from conducting inference using the `exact' posterior to a posterior based on a certain approximate likelihood that we will propose, then we can produce Bayesian credible sets that are asymptotically well-calibrated. To motivate this approach, consider the artificially simple case where the observed data are $Y_1,\dots, Y_n\stackrel{iid}{\sim} P^{ }_0$, with $Y_i\in\mathbb{R}^d$, and we wish to conduct inference on the unknown mean $\E[Y_i]=\theta_\star$.  We have meaningful prior beliefs $\pi(\theta)$ about the unknown $\theta_\star$, and our goal is to conduct posterior inference given $\pi(\theta)$ and  $Y_1,\dots,Y_n$. 

While it is possible to conduct nonparametric Bayesian inference on $\theta$, this seems overly-complicated machinery for such a simple task. A simpler approach is to notice that even though $P_0$ is unknown, the sample average $\overline{Y}=n^{-1}\sum_{i=1}^{n}Y_i$ can be used as a statistic to conduct Bayesian inference on $\theta$. Following the synthetic likelihood (SL) approach proposed by \cite{wood2010statistical}, see also \cite{price2018bayesian} and \cite{frazier2022bayesian}, even though $P^{ }_0$ is unknown we can approximate the distribution of $\overline{Y}\mid\t$ using a known distribution, with this approximation then used as our likelihood to produce posterior inference. 

For general summary statistics $S_n$, \cite{wood2010statistical} suggests approximating the distribution $S_n\mid\t$ using $N\{S_n;b(\theta),W_n(\theta)/n\}$ where $b(\theta)$, and $W_n(\theta)/n$ are the mean and (scaled) variance of $S_n\mid\t$ under the assumed model $P^{(n)}_\t$.\footnote{In the case of the summaries $\overline{Y}$,  SL suggests conducting Bayesian inference on $\theta$ using the approximate likelihood $N\{\overline{Y};\theta,W_n(\theta)/n\}$ where $W_n(\theta)$ is an estimator of $\text{Cov}(\sqrt{n}\overline{Y})$.} %We can then update our prior information $\pi(\t)$ using the likelihood $ N\{\overline{Y};\theta,W_n(\theta)/n\}$ to produce the approximate posterior 
%\begin{flalign*}
%	\pi_n(\theta)&\propto|W_n(\theta)|^{-1/2}\exp\left\{-\frac{1}{2}[\sqrt{n}(\theta-\overline{Y})]^\top W_n(\theta)^{-1}[\sqrt{n}(\theta-\overline{Y})]\right\} \pi(\theta).
%\end{flalign*}
%Under weak conditions,  $\pi_n(\theta)$ has central credible sets that asymptotically behave like
%$
%\{\theta\in\pi(\Theta):\theta\in[\overline{Y}\pm z_{\alpha/2}W_n(\theta_\star)^{-1/2}/\sqrt{n}]\}
%$, where $z_\alpha$ is the $\alpha$-quantile of the standard normal. So long as $W_n(\theta_\star)$ is a consistent estimator of $\text{Cov}(\sqrt{n}\overline{Y})$, which can be achieved by setting $W_n(\theta)$ as the sample variance, $\forall\theta\in\Theta$, $\pi_n(\t)$ will produce Bayesian inferences for $\theta$ that are well-calibrated (asymptotically). 
While Bayesian SL (BSL) is most commonly applied to intractable likelihood problems, BSL can be applied in any situation where we wish to conduct posterior inference on summaries rather than the full dataset; for further discussion on this point see \cite{drovandi2021Rbayesian}. As highlighted by \cite{lewis2021bayesian}, in misspecified models it may be particularly beneficial to condition posterior inferences not on the entire sample, encapsulated via $\ell_n(\theta)$, but on summary statistics that are `robust' to model misspecification. Indeed, when the model is misspecified the exact posterior is not `robust' from the standpoint of uncertainty quantification, and it may instead be feasible to produce Bayesian inferences based on a vector of summaries that ensure our posterior inference for $\theta$ is `robust' to this particular form of misspecification. 

Noting that the average score  $n^{-1}m_n(\theta)$ is simply a summary statistic whose mean is zero under the assumed model (under weak regularity conditions),\footnote{For $P^{(n)}_\theta$ denoting the assumed model,  $P^{(n)}_\theta m_n(\theta) =0$ for all $\theta\in\Theta$, so long as $\nabla_\theta P^{(n)}_\theta\ell_n(\theta)= P^{(n)}_\theta \nabla_\theta\ell_n(\theta)$, for all $\theta\in\Theta$. {We recall that the exchange of integration and differentiation is valid under the following conditions: 1)$\ell_n(\theta)$ is an integrable function of $\y$ with respect to $P^{(n)}_\t$; 2) For almost all $\y\in\mathcal{Y}^{(n)}$, $\nabla_\t\ell_n(\theta)$ exists for all $\theta\in\Theta$; 3) There is an integrable function $g(\y)$ such that $|\nabla_\t\ell_n(\t)|\le g(\y)$ for all $\theta\in\Theta$ and almost all $\y\in\mathcal{Y}^{(n)}$.}} we can apply the SL logic to this setting by constructing a matrix $W_{n}(\theta)$ that is a consistent estimator of $\text{Cov}\{m_n(\theta)/\sqrt{n}\}$. In particular, viewing $n^{-1}m_n(\theta)$ as a centred summary statistic, we can follow \cite{wood2010statistical} and approximate the distribution of $n^{-1}m_n(\theta)\mid\t$ using a mean-zero Gaussian distribution with variance $n^{-1}W_n(\theta)$, which produces the following BSL posterior
\begin{equation}\label{eq:genpost}
	\pi^Q_{n}(\theta)=\frac{M_n(\theta)^{-1/2}\exp\{- Q_n(\theta)\}\pi(\theta)}{\int_\Theta M_n(\theta)^{-1/2}\exp\{- Q_n(\theta)\}\pi(\theta)\dt\theta},	\; Q_n(\theta)=\frac{1}{2}\frac{m_n(\theta)^\top}{\sqrt{n}}W_{n}(\theta)^{-1}\frac{m_n(\theta)}{\sqrt{n}},
\end{equation}	
where $M_n(\theta)=|W_{n}(\theta)|$, and where the notation $\pi^Q_n(\theta)$ encodes the posteriors dependence on $Q_n(\theta)$.\footnote{This distributional approximation takes $n^{-1}m_n(\theta)|\t$ as mean-zero Gaussian with variance $n^{-1}W_n(\theta)$. The term $Q_n(\theta)$ in \eqref{eq:genpost} follows by re-arranging $\{n^{-1}m_n(\theta)\}^\top \{n^{-1}W_n(\theta)\}^{-1}\{n^{-1}m_n(\theta)\}=\{n^{-1/2}m_n(\theta)\}^\top W_n(\theta)^{-1}\{n^{-1/2}m_n(\theta)\}$.}  Critically, unlike standard BSL methods, the posterior in \eqref{eq:genpost} will not lead to a loss in information since, asymptotically, the information in $\ell_n(\theta)$ is the same as in $m_n(\theta)$. 

 While BSL is often applied to intractable models, in our case the exact posterior $\pi(\t\mid\y)$ is tractable, but we replace the information in the full sample  with that in the statistic $W_n(\theta)^{-1/2}m_n(\theta)$ in the hopes that this statistic delivers posterior inferences that are robust to the model misspecification. The posterior  $\pi^Q_{n}(\theta)$ amounts to replacing the usual likelihood $\ell_n(\theta)$, with the information in $Q_n(\theta)$, which resembles a type of (random) quadratic form in the scores $m_n(\t)$. Hence, to differentiate the usual BSL posterior from the posterior in \eqref{eq:genpost}, throughout the remainder we refer to $\pi^Q_n(\theta)$ as a Q-posterior.

	While $Q_n(\theta)$ resembles a quadratic approximation of the log-likelihood, this does not imply that $\pi^Q_n(\theta)$ will resemble a Gaussian distribution.  More generally, $\pi^Q_n(\theta)$ remains meaningful when the parameters are defined over a restricted space, and when the posterior places mass near the boundary of the support. In such cases, a Gaussian approximation to the posterior will not be particularly meaningful. In this way, the Q-posterior differs from existing approaches used to produce well-calibrated Bayesian inference, as it does not rely on post-processing of the accepted draws, or on additional replications of the sampling algorithm. As argued more explicitly in Section \ref{sec:discuss}, post-processing steps are generally based on an implicit, or explicit, normality assumption and thus are not meaningful in cases where the parameter has restricted support. In addition, unlike frequentest methods that seek to correctly quantify uncertainty, the Q-posterior correctly quantifies uncertainty without needing to estimate the second derivative matrix $\mathcal{H}_n(\theta)$.

The following remarks relate the Q-posterior with several existing approaches in the literature. 
\begin{remark}\normalfont The posterior $\pi^Q_n(\theta)$ resembles a form of the generalized posterior of \cite{bissiri2016general} under a (weighted) quadratic loss. In particular, if we have $m_i(\theta)=\nabla_\theta \log p(y_i\mid\theta)$,  $m_{n,i}(\theta)=\{W_{n}(\theta)/n\}^{-1/2}m_i(\theta)$, and $q_{n,i}(\theta)=-\frac{1}{2}\|m_{n,i}(\theta)\|^2$, we have
$
\exp\{-Q_n(\theta)\}=\exp\{\sum_{i=1}^{n}q_{n,i}(\theta)\}$	which resembles the generalized posterior based on a (weighted) quadratic loss function. However, our approach is motivated by attempting to produce posteriors with appropriate uncertainty quantification, and not the use of other loss functions within Bayesian inference (see Section \ref{sec:general} for discussion). Moreover, unlike existing generalized Bayes methods the Q-posterior does not require the difficult choice of tuning constant, which can greatly impact the posterior uncertainty quantification produced using generalized Bayesian methods. 
\end{remark}
\begin{remark}\normalfont 
\cite{chernozhukov2003mcmc} propose a type of generalized Bayesian inference based on either a set of over-identified estimation equations for $\theta$, more equations than unknown parameters, by taking a quadratic form of a vector of sample moments (see, also, \citealp{chib2018bayesian}), or by replacing the likelihood altogether with an $M$-estimator criterion; the latter is also used in a decision theoretic framework by \citet{bissiri2016general} to produce their generalized posterior. However, neither approach covers the case where the loss function defining the posterior is based on the score equations from the likelihood. Philosophically, the approaches of \cite{chernozhukov2003mcmc} and \cite{bissiri2016general} are based on conducting a form of Bayesian inference \textit{by bypassing the likelihood}, and not on conducting Bayesian inference using the likelihood in misspecified models. Therefore, from a philosophical standpoint, the two are distinct. {Structurally, the two are also distinct since the posterior in \eqref{eq:genpost} and that proposed by the aforementioned authors are of different, but related, forms.}
\end{remark}

\begin{example}[Linear Regression]\label{ex:reg}
	\normalfont
%To demonstrate the benefits of our generalized approach to Bayesian inference over standard Bayesian inference, we demonstrate the conclusions of Lemmas \ref{lem:two}-\ref{lem:three} in a simple example. 
Consider the standard linear regression model
$$
y_{i}=x_i^\top\beta+\sigma\epsilon_i,\;(i=1,\dots,n),
$$where the error distribution is assumed to be $\epsilon_i\stackrel{iid}{\sim} N(0,1)$, and where $x_i$ and $\beta$ are $d_\beta$-dimensional vectors. For simplicity, we consider diffuse priors for these parameters: $\pi(\beta)\propto1$ and $\pi(\sigma)\propto (\sigma^2)^{-2}$. While we believe the model specification for the regression components to be correct, we are uncertain about the error specification. The true data generating process produces observed data under the following specification for the error term:
$
\epsilon_i\stackrel{iid}{\sim}N(0,\sqrt{1+|x_{1,i}|^\gamma}^2),
$ where $\gamma\ge0$, and $x_{1,i}$ denotes the first element of the vector $x_i$. When $\gamma=0$ the assumed model with homoskedastic errors is correctly specified, whereas if $\gamma\ne0$, the assumed model is misspecified and the exact posterior will likely produce credible sets that are too narrow. %Alternatively, from Lemmas \ref{lem:two}-\ref{lem:three} suggest that the generalized posterior we propose will correctly quantify uncertainty regardless of the value of $\gamma$. 

We consider generating $n=100$ observations from the model under the designs of $\gamma=0$, and $\gamma=2$;  $x_i$ is generated as tri-variate independent standard Gaussian random vectors, and we set the true values  as $\beta =(1,1,1)^\top$, and $\sigma =1$. When the model is misspecified it can be shown that the pseudo-true value of $\beta$ is unchanged and so we focus in this example only on inferences for $\beta$. We replicate this design 1000 times, and for each dataset sample the exact posteriors using Gibbs sampling, and the Q-posterior is sampled using random walk Metropolis-Hastings (RWMH) with a Gaussian proposal kernel; both posteriors are approximated using 10000 samples with the first 5000 discarded for burn-in. We set the matrix $W_n(\theta)$ in the Q-posterior to be $W_n(\theta)=n^{-1}\sum_{i=1}^{n}\{m_i(\theta)-n^{-1}{m}_n(\theta)\}\{m_i(\theta)-n^{-1}{m}_n(\theta)\}^\top$. Across each dataset and method, we compare the posterior bias, variance and marginal coverage for the regression coefficients. 

Table \ref{tab:reg} summarizes the results, and demonstrates that the Q-posterior produces results that are similarly located to exact Bayes, but which have larger posterior variance under both regimes. This larger posterior variance produces coverage rates that are close to the nominal 95\% level, and much more so than under the exact posterior. Even in the homoskedastic regime ($\gamma=0$), the exact posterior has converge rates around 85\%, while in the heteroskedastic regime ($\gamma=2$), the lowest level of coverage is less than 70\%.

	\begin{table}[H]
	\centering
	{\footnotesize
		\begin{tabular}{lrrrrrr}
			\hline\hline
			\textbf{$\gamma=0$} & \multicolumn{3}{c}{Q-posterior} & \multicolumn{3}{c}{Exact} \\
			\textbf{} & Bias & Var &Cov& Bias & Var &Cov \\ \hline
$\beta_1$ &-0.0009  &  0.0395 &   0.9720  & -0.0110 &   0.0104 &   0.8380\\
$\beta_2$ & 0.0188  &  0.0388  &  0.9680  & -0.0049 &   0.0103 &   0.8370\\
$\beta_3$ &-0.0212 &   0.0408  &  0.9710  & -0.0103 &   0.0105 &   0.8540\\\hline\hline
			\textbf{$\gamma=2$}& \multicolumn{3}{c}{Q-posterior} & \multicolumn{3}{c}{Exact}   \\ 
			\textbf{} & Bias & Var & Cov& Bias & Var &Cov \\ \hline
$\beta_1$ & -0.0072 &   0.0520  &  0.9590 &  -0.0215 &   0.0104 &   0.6860\\
$\beta_2$ & -0.0265  &  0.0343  &  0.9660 &  -0.0196  &  0.0104  &  0.8380\\
$\beta_3$ & -0.0027 &   0.0277  &  0.9680 &  -0.0082 &   0.0104 &   0.8440\\\hline\hline
	\end{tabular}}
	\caption{Posterior accuracy results in the normal linear regression model for the exact and Q-posterior.  Bias is the bias of the posterior mean across the replications. Var is the average posterior variance deviation across the replications. Cov is the posterior coverage. The nominal level is set to be 95\% for the experiments. }
	\label{tab:reg}
\end{table}

\end{example}
%\end{remark}

\subsection{Conjugacy for Exponential Family Models in Natural Form}
Consider the case where $Y_1,\dots,Y_n\stackrel{iid}{\sim}p_\t(y)=\exp\{\eta(\theta)^\top S(y)-A(\t)\}h(y),
$ where $S:\mathcal{Y}\rightarrow\mathbb{R}^{d_\t}$ is a vector of sufficient statistics,  $h:\mathcal{Y}\rightarrow \mathbb{R}$ a reference measure on $\mathcal{Y}$, and $A:\Theta\rightarrow\mathbb{R}$ the log-partition function. Then, the joint density $p_\t^{(n)}(\y)$ takes the form 
$$
p_\t^{(n)}(\y)=\exp\left\{\eta(\theta)^\top \sum_{i=1}^{n}S(y_i)-nA(\t)\right\}\prod_{i=1}^{n}h(y_i),
$$where $
A(\theta)=\log \left[\int \exp \left\{\eta(\theta)^\top S(x)\right\} h(x) \dt \mu(x)\right]
$.

Consider that our goal is inference on the natural parameters $\eta=\eta(\t)$. Then, we can show that the Q-posterior has a closed-form expression so long as our prior for $\eta$ is conjugate. In particular, 
it is simple to see that the above model has average scores $n^{-1}m_n(\eta)=\nabla_\eta A(\eta)-n^{-1}S_n$, where $S_n=\sum_{i=1}^{n}S(y_i)$. Since $A(\eta)$ is non-random, the variance of $n^{-1}m_n(\eta)$ can be estimated using the sample variance
$
W_n:=\frac{1}{n}\sum_{i=1}^{n}\left\{S(y_i)-n^{-1}S_n\right\} \left\{S(y_i)-n^{-1}S_n\right\} ^\top,
$ which does not depend on $\eta$.\footnote{More generally, any consistent estimator of $\text{Var}\{S(y_i)\}$ could also be used.} One could then consider inference on $\eta$ using $\pi^Q_n(\eta)\propto \exp\{-Q_n(\eta)\}\pi(\eta)$, where
$
Q_n(\eta)=\frac{n}{2}\left\{\nabla_\eta A(\eta)-n^{-1}S_n\right\}^\top W_n^{-1}\{\nabla_\eta A(\eta)-n^{-1}S_n\}.
$

Define the mean form parameter $\mu$ by the function $\mu=g(\eta)=\nabla_\eta A(\eta)$. In regular models the function $g(\eta)$ exists and is invertible for all $\eta$. The parameter $\mu=\mu(\eta)$ is referred to as the mean parameterization of the model and satisfies $\mu=\E_{Y\sim p^{(n)}_\t(\y) }[S_n]$. The form of the Q-posterior for $\eta$ then follows by finding the Q-posterior for $\mu$, and using the change of variables formula. 

\begin{lemma}\label{lem:conjugate}
Consider that $\mathcal{Y}=\mathbb{R}^d$. If our (transformed) prior beliefs for the mean parameter $\mu=g(\eta)$ can be written as $\pi(\mu)\propto \exp\{-\frac{1}{2}(\mu-\mu_0)^\top W_0^{-1}(\mu-\mu_0)\}$, then the Q-posterior for $\eta$ is
$
\pi^Q_n(\eta)=N\{g(\eta);b_n,\Sigma^{-1}_n\}|\nabla_\eta^2 A(\eta)|
$, where 
\begin{flalign*}
\Sigma^{-1}_n&=n^{-1}W_0\left[n^{-1}W_n+W_0\right]^{-1}W_n,\\ b_n&=W_0\left[n^{-1}W_n+W_0\right]^{-1}\frac{1}{n}\sum_{i=1}^{n}S(y_i)+W_0\left[n^{-1}W_n+W_0\right]^{-1}\frac{\mu_0}{n}.	
\end{flalign*}
\end{lemma}

Lemma \ref{lem:conjugate} demonstrates that for exponential families in normal form, the Q-posterior for the natural parameters is a transformed Gaussian density.\footnote{So long as the (transformed) prior under the mean parameterization, $\pi(\mu)$, can be written as a Gaussian kernel.} Interestingly, calculation of $|\nabla_\eta^2 A(\eta)|$ can be entirely avoided by first sampling $\widetilde\mu\sim N\{\mu;b_n,\Sigma_n^{-1}\}$, and then (numerically) inverting the equation $\widetilde\mu=g(\eta)$ to obtain the draw $\widetilde\eta$. The latter can be carried out in any case where $g(\eta)=\nabla_\eta A(\eta)$ can be reliably calculated. 

%The critical feature of exponential families that allows for the closed form posteriors in Lemma \ref{lem:conjugate}, is the fact that the variance matrix $W_n$ does not depend on $\eta$, and is, thus, not impacted by possible misspecification of the model. Note that, if we were to consider the variance `under the model', we could use the log-partition function; which yields $\Cov\{S(y_i)|\eta\}=\nabla_\eta^2A(\eta)$. Of course, if the model is misspecified, there is no reason to believe that $\Cov\{S(y_i)|\eta\}=\nabla_\eta^2A(\eta)$ will be a consistent estimator of $\Cov\{S(y_i)\}$. Such an issue is precisely the reason behind the poor coverage rates in the simple linear regression example: the exact posterior assumes that $\Cov\{S(y_i)|\eta\}$ is constant, as a function of $\eta$, whereas the Q-posterior does not. Hence, when the model is misspecified, the exact posterior does not adapt to the increased variability brought by the heteroskedastic nature of the variance, while the Q-posterior adapts to the variability in the observed data.  

\section{Estimated Likelihoods}\label{sec:ext}
The Q-posterior $\pi_n^Q(\theta)$ in \eqref{eq:genpost} requires that $m_n(\theta)=\nabla_\theta\ell_n(\theta)$, and $W_n(\theta)$ are known up to $\theta$. In many interesting cases, however, the observed data likelihood is only expressible as
$$
p_\theta^{(n)}(\y)=\int_{\mathcal{A}}p_\theta^{(n)}(\y,\a)\dt\a,
$$where $p_\theta^{(n)}(\y,\a)$ is called the complete-data likelihood, $\a\sim p(\a\mid\theta)$, and $\alpha_i\in\mathcal{A}\subseteq\mathbb{R}^{d_\alpha}$ ($i=1,\dots,n$) are unobservable. In such cases, $p_\theta^{(n)}(\y)$ is often intractable, and we cannot obtain analytic, or computationally tractable, forms for $m_n(\theta)$ and $W_n(\theta)$. Further, while it may be feasible to obtain an estimator of the likelihood, say $\widehat{p}_\theta(\y\mid z)$, differentiating $\widehat{p}_\theta(\y\mid z)$ to obtain an estimator for $m_n(\theta)$ may be undesirable; in many cases, $\widehat{p}_\theta(\y\mid z)$ may not be differentiable in $\theta$, as is the case with particle filtering methods. Consequently, we must somehow approximate $m_n(\theta)$ to extend the Q-posterior to cases where $p_\theta^{(n)}(\y)$ must be estimated.

Fisher's identity (\citealp{cappe2005springer}) is a relationship between the gradient of the observed log-likelihood, $\log p_\theta^{(n)}(\y)$, and the complete data log-likelihood $p_\theta(\y,\a)$, which allows us to  express $m_n(\theta)=\nabla_\theta \log p_\theta^{(n)} (y_{1:n})$ as
\begin{equation*}
	m_n(\theta)=-\nabla_\theta \ell_n(\theta)=\int \left\{\nabla_\theta\log 	p_{\theta}\left(\y, \a\right)\right\} p_\theta(\a\mid y_{1:n})\dt \a,
\end{equation*}where $p_\theta(\a\mid\y)$ is the posterior for $\a$ conditional on $\theta$. Given draws $\a^{(j)}\stackrel{iid}{\sim}p_\t(\a\mid y_{1:n})$ ($j=1,\dots,N$), we can estimate $m_n(\theta)$ using the simple Monte Carlo estimator
$$
\widehat{m}_n(\theta;z):=\frac{1}{N}\sum_{j=1}^{N}\nabla_\theta \log p_\theta(\y,\a^{(j)}),\quad \a^{(j)}\sim p_\theta(\a\mid\y),
$$ where, again, $z$ denotes all simulated variables necessary to construct $\widehat{m}_n(\theta;z)$.\footnote{For notational simplicity, we disregard the estimators dependence on $N$, as we will later take $N$ as a function of $n$.} All we require to obtain $\widehat{m}_n(\theta;z)$, and by extension an estimator of $W_n(\t)$, is that we can generate samples from $p_\t(\a\mid\y)$.\footnote{While more efficient estimators of $m_n(\theta)$ can be obtained, such as those based on importance sampling, the simple Monte Carlo estimator is both effective and theoretically convenient to analyze. We conjecture that more efficient estimators can be used, but leave a formal study of such estimators for future research.}

The ability to generate draws from $p_\t(\a\mid\y)$ is not particularly restrictive, as we do not require that $p_\t(\a\mid\y)$ be available in closed-form. As the following examples demonstrate, this is feasible in many cases, such as, e.g., state-space models, and generalized random effects models.

\begin{example}\label{rem:fisher}\normalfont
In   state space models  the complete-data likelihood $p_\t(\a,\y)$ can often be written as
	\begin{equation*} 
		p_{\t}\left(\a, \y\right)=\mu_{}\left(\alpha_{1}\right) \prod_{t=2}^{n} f_\t\left(\alpha_{t} \mid \alpha_{t-1}\right) \prod_{t=2}^{n} g_{\t}\left(y_{t} \mid \alpha_{t}\right),
	\end{equation*}where $f_\theta(\alpha\mid\alpha')$ is the transition density of some unobservable (state) process $\{\alpha_t:t\le n\}$, $g_{\theta}\left(y_{t} \mid \alpha_{t}\right)$ is the conditional density of $y_t\mid \alpha_t$, and $\mu(\alpha)$ is the invariant measure of the states, and where the process $\{y_t:t\le n\}$ is conditionally independent given $\{\alpha_t: t\le n\}$.

	The Markovian nature of the states implies that filtering and smoothing methods can be used to produce draws from $p_\t(\a\mid\y)$. In addition, there are many situations, such as linear-Gaussian models and some stochastic volatility models, where draws from  $p_\t(\a\mid\y)$ can be obtained using Gibbs or MH sampling algorithms. Regardless of how draws from $p_\t(\a\mid\y)$ are obtained, $m_n(\theta)$ can be estimated using
	\begin{flalign}\label{eq:grad_est}
	\widehat{m}_n(\theta;z)&=N^{-1}\sum_{j=1}^{N}\sum_{t=2}^{n}\nabla_\theta \log g_\t(y_t\mid \alpha^{(j)}_t)+N^{-1}\sum_{j=1}^{N}\sum_{t=2}^{n}\nabla_\theta \log f_\t(\alpha_t^{(j)}\mid \alpha^{(j)}_{t-1})\nonumber,
\end{flalign} for $\a^{(j)}\sim p_\t(\a\mid\y)$. Further, given values of $\alpha$ and $\alpha'$, $\nabla\log f_\theta(\alpha\mid\alpha')$ and $\nabla_\theta\log g_\theta(y\mid \alpha)$ can often be calculated in closed-form.
\end{example}
\begin{example}\label{rem:IS} \normalfont
Random effects models are ubiquitous in the analysis of clustered and longitudinal data. In such models, the outcome variable, $y_{i,j}$, $i=1,\dots,n$ and $j=1,\dots,J$, is influenced by a vector of covariates $x_{i,j}$ and an unobservable random effect $\alpha_i$. In most cases, conditional on a scale parameter $\tau$, the random effects are assumed to be iid, so that their joint distribution is $f_\tau(\a)=\prod_{i=1}^{n}f_\tau(\alpha_i)$, with $f_\tau(\cdot)$ a mean-zero distribution, known up to $\tau$. Conditional on $\alpha_{1:n}$ and $x_{1:n,j}$, the distribution of the outcomes is $g_\theta(y_{1:n,j}\mid \a, x_{1:n,j})=\prod_{i=1}^{n}g_\theta(y_{i,j}\mid \alpha_i, x_{i,j})$, where $g_\t(y_{i,j}\mid \alpha_i, x_{i,j})$ is a member of the exponential family that depends on the linear index $\eta_{i,j}=\alpha_i+x_{i,j}^\top\beta$, and where $\beta$ is a vector of regression parameters, so that $\theta=(\beta^\top,\tau)^\top$. The complete-data likelihood again has a product form, and given posterior draws from $\a\mid y_{1:n,1;J},x_{1:n,1:J},\theta$, the estimator $\widehat{m}_n(\theta;z)$ can be constructed as before.

\end{example}
Fisher's identity and simple Monte Carlo allow us to approximate $m_n(\theta)$ by averaging $N$ draws from $p_\theta(\a\mid\y)$, where, for each $j=1,\dots,N$, the $j$-th draw depends on random variables $z_j$. Concatenating all such random variables as $z$, we denote the simulated estimator of the score equations, $m_n(\theta)$, as $\widehat{m}_n(\theta;z)$, and let $\widehat{W}_n(\theta;z)$ denote an estimator of the variance of $\widehat{m}_n(\t;z)$. These functions dependence on $z$ is similar to the pseudo-marginal literature, where the likelihood estimator $\widehat{p}_\theta(\y\mid z)$ is also viewed as dependent on simulated variables $z$.

Given $\widehat{m}_n(\theta;z)$, and $\widehat{W}_n(\theta;z)$, posterior inferences for $\theta$ can proceed using, e.g., a RWMH-MCMC scheme that replaces the intractable $Q_n(\theta)$ with the estimated version
$$
\widehat{Q}_n(\theta;z)=\frac{1}{2}\widehat{m}_n(\theta;z)^\top \widehat{W}_n(\theta;z)^{-1} \widehat{m}_n(\theta;z).
$$ See Algorithm \ref{alg1}. %To demonstrate the applicability of this approach, we apply this approach in two examples: a random effects model, and a time varying vector-autoregressive models with stochastic volatility.   
\begin{algorithm}[h]
	\caption{Pseudo-marginal MH-MCMC with Q function (MH-MCMC-Q)}
	\label{alg1}
	\begin{algorithmic}
		\STATE Initialize $\theta^{(0)}$ and ${z}^{(0)}$.
		\FOR{$i=1,\dots,M$}	
		\STATE Draw $\theta^\ast\sim q(\theta|\theta^{i-1})$
		\STATE Draw $z^\ast$ from the proposal $h(z\mid\theta^\ast)$
		\STATE Compute  $V^\ast=\widehat{g}_n(\theta^\ast;z^\ast):=\widehat{M}_n(\theta;z)^{-1/2}\exp\{-\widehat{Q}_n(\theta^\ast;z^\ast)\}$ and $V^{i-1}=\widehat{g}_n(\theta^{i-1};z^{i-1})$
		\STATE Compute the Metropolis-Hastings ratio:
		$
		r = \alpha(\theta^\ast,\theta^{(i-1)})=\frac{V^\ast\pi(\theta^\ast)q(\theta^{i-1}|\theta^\ast)}{V^{i-1}\pi(\theta^{i-1})q(\theta^\ast|\theta^{i-1})}
		$
		\IF{$\mathcal{U}(0,1)<r$}
		\STATE Set $\theta^i=\theta^\ast$, $z^i= z^\ast$	
		\ELSE
		\STATE Set $\theta^i=\theta^{i-1}$, $z^i=z^{i-1}$	
		\ENDIF
		\ENDFOR 
	\end{algorithmic}	
\end{algorithm}

\subsection{Posterior Target, and Behavior in Large Samples}\label{sec:large}
 The need to estimate $Q_n(\theta)$, via $\widehat{Q}_n(\theta;z)$, and the replacement of the intractable ``likelihood term'' $\exp\{-Q_n(\theta)\}$ within the MCMC algorithm by $\exp\{-\widehat{Q}_n(\theta;z)\}$ results in an algorithm that does not deliver draws from $\pi^Q_n(\theta)$, the posterior we would like to target if $m_n(\theta)$ were tractable. Following the results of \cite{andrieu2009pseudo}, Algorithm \ref{alg1} can be viewed as an exact MH algorithm on $\Theta\times\mathcal{Z}$ with stationary distribution 
$$
\overline{\pi}^Q_{n}(\theta,z)\propto {\widehat{g}_{n}(\theta;z)\pi(\theta)h(z\mid\theta)},\quad \widehat{g}_{n}(\theta;z):=\widehat{M}_n(\theta;z)^{-1/2}\exp\{-\widehat{Q}_n(\theta;z)\},
$$where $\widehat{M}_n(\theta;z)=|\widehat{W}_n(\t;z)|$. Therefore, the marginal posterior for $\theta$ from Algorithm \ref{alg1} is
$$
\overline{\pi}^Q_{n}(\theta)\propto \int_{\mathcal{Z}}\widehat{g}_{n}(\theta;z)\pi(\theta)h(z\mid\theta)\dt z. $$
Unlike standard pseudo-marginal methods, $\widehat{g}_n(\theta;z)$ is a biased estimator of the posterior kernel $M_n(\theta)^{-1/2}\exp\{-Q_n(\theta)\}$. Consequently, the use of  $\widehat{g}_n(\theta;z)$ within an MCMC algorithm results in a posterior that does not agree with $\pi^Q_n(\theta)$ in \eqref{eq:genpost}. 

The posterior $\overline{\pi}^Q_{n}(\theta)$ is only expressible as an intractable integral, so that deriving anything concrete about its behavior for fixed $n$ and $N$ seems infeasible. However, so long as $N\rightarrow+\infty$ we can ensure that $\overline{\pi}^Q_{n}(\theta)$ is uniformly close to $\pi^Q_n(\theta)$, as $N\rightarrow+\infty$ for any $n\ge1$.  

\subsubsection{Main Result}
To control the behavior of $\overline\pi^Q_n(\theta)$ we must control the behavior of $\widehat{Q}_n(\theta;z)$. %We remind the reader that $\widehat{W}_n(\theta;z)$ is an estimator of $W_n(\theta)$, the variance of the infeasible ${m}_n(\theta)$, and which is itself based on $N$ simulated variables $z=(z_1^\top,\dots,z_N^\top)^\top$. 
We achieve this control  under certain moment assumptions on $\widehat{m}_n(\theta;z)$ and $\widehat{W}_n(\theta;z)$.

\begin{assumption}\label{ass:tails}There exists a function $\sigma_n(\theta):\Theta\rightarrow\mathbb{R}_{+}$ such that the following are satisfied for all $n\ge1$ and for $j=2,4$. 
	\begin{enumerate}
		\item 	$\E_{z}[\|\{\widehat{m}_n(\theta;z)-m_n(\theta)\}/\sqrt{n}\|^{j}]\le \sigma^j_n(\theta)/N$. 
		\item 
$\E_{z}[\|\widehat{W}_n(\theta;z)^{}-W_n(\theta)^{}\|^{j}]\le \sigma^j_n(\theta)/N$.
		\item  For $\|\cdot\|$ an appropriate matrix norm, some $C>0$, $N$ large enough, and all $\theta\in\Theta$, $\|W_n(\theta)^{-1}\{\widehat{W}_n(\theta;z)-W_n(\theta)\}\|\le C/N$, and   $\sup_{\theta\in\Theta}\|\sigma_n^2(\theta)W_n(\theta)^{-1}\|<C$. 
	\end{enumerate}

\end{assumption}

Assumption \ref{ass:tails} restricts the behavior of the estimators used for $m_n(\theta)$ and $W_n(\theta)$. Assumption \ref{ass:tails}(1-2) require that the simulation-based estimators have second moments in $z$ that are bounded by some function $\sigma^{2}_n(\theta)$. Assumption \ref{ass:tails}(3) gives a relationship between the infeasible $W_n(\theta)$, i.e., covariance of the scores, and the function that bounds their moments, $\sigma^2_n(\t)$. Such an assumption is similar to those required when carrying out perturbation analyses for the solution of linear systems (Chapter 5.8 in \citealp{horn2012matrix}). Together, these  restrictions allow us to upper bound the mean and variance of $\widehat{Q}_n(\theta;z)$. As we are unaware of a general result that bounds the mean and variance of a quadratic form with a random weighting matrix, we provide such a general result as Lemma \ref{lem:LIE} in Appendix \ref{app:Lemmas}, which may itself be of independent interest. %These bounds allow us to bound  $\E_{z}\widehat{g}_n(\theta,z)$ and control the behavior $\bar\pi^Q_n(\theta)$. 

\begin{theorem}\label{thm:main}For all $n\ge1$ assume that $\pi_n^Q(\theta)$ exists, and for some known function $\varphi:\Theta\rightarrow\mathbb{R}^{d_\varphi}$, assume that $\int_\Theta\|\varphi(\theta)\|\pi_n^Q(\theta)\dt\t<\infty$, and 	$\int_\Theta \|\varphi(\theta)\|\sigma^4_n(\theta)\pi(\theta)\dt\theta<\infty$. If Assumption \ref{ass:tails} is satisfied, then as $N\rightarrow+\infty$
	$$
	\int_\Theta |\pi^Q_n(\theta)-\overline\pi^Q_n(\theta)|\dt\t=O(1/N),\text{ and }	\left\|\int_\Theta \varphi(\theta)\pi^Q_n(\theta)\dt\t-\int_\Theta\varphi(\theta)\overline\pi^Q_n(\theta)\dt\t\right\|=O(1/N).
	$$  
\end{theorem}	
Theorem \ref{thm:main} demonstrates that, for $N\rightarrow+\infty$, the difference between the intractable posterior $\pi^Q_n(\theta)$ and its `estimated version' $\overline\pi^Q_n(\theta)$, converges to zero at rate $1/N$. This result is not an `in probability' result, and is true for any fixed $n$, as $N\rightarrow+\infty$. Furthermore, this equivalence holds for any function of $\theta$ such that $\|\int_\Theta \varphi(\theta)\pi^Q_n(\theta)\dt\theta\|<\infty$: posterior moments calculated from $\overline\pi^Q_n(\theta)$ are equivalent to those based on the intractable  $\pi^Q_n(\theta)$, at order $O(1/N)$. Thus, for $N$ large, Theorem \ref{thm:main} implies that  $\overline\pi^Q_n(\theta)$, and its moments,  $\int_\Theta\varphi(\theta)\overline\pi^Q_n(\theta)\dt\t$, behave precisely as if $m_n(\theta)$ and $W_n(\theta)$ were known. 

Intuitively, Theorem \ref{thm:main} implies that for $N$ large, inferences based on the feasible posterior $\overline{\pi}^Q_{n}(\theta)$, are uniformly close to those of the infeasible posterior ${\pi}^Q_{n}(\theta)$ that we would hope to target if the likelihood was analytically tractable. An immediate implication of Theorem \ref{thm:main} is that if $\pi^Q_n(\theta)$ has credible sets that are well-calibrated, then credible sets based on $\overline{\pi}^Q_{n}(\theta)$ will also be well-calibrated, even though the likelihood must be estimated to feasibly conduct inference. 

\subsubsection{Uncertainty Quantification}
If we are willing to make additional assumptions regarding the behavior of $m_n(\theta)$ and $W_n(\theta)$, we can formally demonstrate that $\overline{\pi}^Q_n(\theta)$ and $\pi^Q_n(\theta)$ have credible sets that are asymptotically well-calibrated, regardless of whether the model is, or is not, correctly specified. Importantly, unlike frequentist methods for models with latent variables, credible sets obtained from $\overline\pi^Q_n(\theta)$ \textit{do not require the calculation of second-derivative information.} %All that we require is that $n^{-1}\widehat{W}_n(\theta_\star;z)$ be a consistent estimator of $\lim_n\text{Cov}\{m_n(\theta_\star)/\sqrt{n}\}$, which frequentist methods would also require if they intended to produce valid confidence intervals. 

%\paragraph{Notations}
Before presenting our regularity conditions and results, we collect here notations that are used to simplify the statement of the result. For a positive sequence $a_n\rightarrow+\infty$ as $n\rightarrow+\infty$, we say that $X_n=o_p(a_n^{-1})$ if the sequence $a_n X_n$ converges to zero in probability, while we use the notation $X_n=O_p(a_n^{-1})$ to denote that $a_nX_n$ is bounded in probability. For a set $A\subseteq\mathbb{R}^d$, let $\mathrm{Int}(A)$ denote the interior of $A$. The notation $\Rightarrow$ denotes weak convergence under $P^{(n)}_0$.

%Since we take $N\rightarrow+\infty$ as $n\rightarrow+\infty$, it is without loss of generality to view $\overline\pi_{n}(\theta)$, as a sequence indexed only in $n$. Hence, throughout the remainder, we do not express the dependence of the posterior $\overline\pi_{n}(\theta)$ on $N$ and only write $\overline\pi_n(\theta)$. 

\begin{assumption}\label{ass:infeasible}There exists a function $m:\Theta\rightarrow\mathbb{R}^{d_\theta}$ such that:
	\begin{enumerate}
		\item $\sup_{\theta\in\Theta}\|n^{-1}m_n(\theta)-m(\theta)\|=o_p(1)$. 
		\item There exist a unique $\theta_\star\in\mathrm{Int}(\Theta)$, such that, $m(\theta)=0$ iff $\theta=\theta_\star$. 
		\item For some $\delta>0$, $m(\theta)$ is continuously differentiable over $\|\theta-\theta_\star\|\le\delta$, and $\mathcal{H}(\theta_\star):=-\nabla_{\theta}m(\theta_\star)$ is invertible.
		\item There exists a positive  semi-definite matrix $W_\star$, such that $m_n(\theta_\star)/\sqrt{n}\Rightarrow N(0,W_\star)$.
		\item For any $\epsilon>0$, there is a $\delta>0$ such that $$P^{(n)}_0\left(\sup _{\left\|\theta-\theta'\right\|<\delta}\frac{\left\|n^{-1}\{m_{n}(\theta)-m_{n}(\theta')\}-\{m(\theta)-m(\theta')\}\right\|}{1+\sqrt{n}\|\theta-\theta'\|}>\epsilon\right)=o(1).$$
	\end{enumerate}
	
\end{assumption}

%\textbf{If we want to include a third order term,and get a non-Gaussian approximation, then we will also need that }
%For any $\epsilon>0$, there is a $\delta>0$ such that $$P^{(n)}_0\left(\sup _{\left\|\theta-\theta'\right\|<\delta}\frac{\left\|n^{-1}\{\nabla_\theta m_{n}(\theta)-\nabla_\theta m_{n}(\theta')\}-\{\nabla_\theta m(\theta)-\nabla_\theta m(\theta')\}\right\|}{1+\sqrt{n}\|\theta-\theta'\|}>\epsilon\right)=o(1).$$

We require that the weighting matrix satisfy the following assumption. 
\begin{assumption}\label{ass:weight}The following conditions
	are satisfied for some $\delta>0$: (i) for $n$ large enough, the matrix ${W}_n(\theta)$ is positive semi-definite and symmetric uniformly over $\Theta$, and positive-definite for all $\|\theta-\theta_\star\|\le\delta$; (ii) there exists some matrix $W(\theta)$, positive semi-definite, and symmetric, uniformly over $\Theta$, and such that $\sup_{\theta\in\Theta}\|{W}^{}_n(\theta)-W(\theta)\|=o_{p}(1)$, and, for all $\|\theta-\theta_\star\|\le\delta$, $W(\theta)$ is continuous and positive-definite;  (iv) $W(\theta_\star)=W_\star$; (v) for any $\epsilon>0$, $\sup_{\|\theta-\theta_\star\|\ge\epsilon}-m(\theta)^\top W(\theta)^{-1}m(\theta)<0$.
\end{assumption}
\begin{remark}{\normalfont 
		Assumptions \ref{ass:infeasible} and \ref{ass:weight} jointly enforce smoothness and identification conditions for the infeasible criterion $Q_n(\theta).$ These conditions permit the existences of a quadratic expansion of $Q_n(\theta)-Q_n(\theta_\star)$ that is smooth in $\theta$ near $\theta_\star$, and with a remainder term that can be suitably controlled. In particular, our assumptions do not require that the loss function is differentiable everywhere, but only that it is equicontinuous. The latter is an important distinction as many loss functions are only differentiable almost surely; e.g., loss functions based on absolute value functions.}
\end{remark}

The following assumption requires the  existence of certain prior moments. 
\begin{assumption}\label{ass:prior}
	For $\theta_\star$ as defined in Assumption \ref{ass:infeasible}, $\pi(\theta_\star)>0$ and $\pi(\theta)$ is continuous on $\Theta$. For some $p\ge1$, $\int_\Theta\|\theta\|^p\pi(\theta)\dt\theta$, and for all $n$ large enough $\int_\Theta |W_n(\theta)|^{-1/2}\|\theta-\theta_\star\|^{p}\pi\left(\theta\right)\dt\theta<+\infty$.
\end{assumption}

%\subsubsection{Large Sample Behavior}
The following result demonstrates that  $\overline\pi_n^Q(\t)$ delivers valid (frequentist) uncertainty quantification. Before stating this result, recall that $\Sigma_\star:=\mathcal{H}(\theta_\star)W_\star^{-1}\mathcal{H}(\theta_\star)$, and define $$Z_n:=\Sigma_\star^{-1}\mathcal{H}(\theta_\star)^\top W_\star^{-1}m_n(\theta_\star),\quad\vartheta:=\sqrt{n}(\theta-\theta_\star)-Z_n/\sqrt{n}$$ and $\mathcal{T}_n:=\{\vartheta=\sqrt{n}(\theta-\theta_\star)-Z_n/\sqrt{n}:\theta\in\Theta\}$.

\begin{lemma}\label{lem:two}Under Assumptions \ref{ass:tails}-\ref{ass:prior}, as $n,N\rightarrow+\infty$, 
	$
	\int_{\mathcal{T}_n}\|\vartheta\||\overline\pi^Q_{n}(\vartheta)-N\{\vartheta;0,\Sigma_\star^{-1}\}|\dt \vartheta=o_p(1).
	$
\end{lemma}
The following result is a consequence of Lemma \ref{lem:two}. 
\begin{lemma}\label{lem:three}Under Assumptions \ref{ass:tails}-\ref{ass:prior}, if $\sqrt{n}/N\rightarrow0$, for $\overline\theta=\int_{\theta}\theta\overline\pi^Q_n(\theta)\dt\theta$, 
	$
	\sqrt{n}\{\overline\theta-\theta_\star\}\Rightarrow N\{0,\Sigma_\star^{-1}\}.
	$
\end{lemma}

Lemmas \ref{lem:two}-\ref{lem:three} demonstrate that even when the likelihood must be estimated, so long as it is feasible to simulate draws from $p_\t(\a\mid\y)$, we can accurately quantify uncertainty using $\overline\pi^Q_n(\theta)$ (for $N\rightarrow+\infty$). However, if we wish for our posterior means to not be asymptotically biased, due to the need to simulate data, we require that $\sqrt{n}/N=o(1)$. The latter condition is the same condition on the number of simulated draws required for asymptotically Gaussian point estimation in the context of simulated maximum likelihood estimation (see \citealp{shephard1997likelihood} for early work). Therefore, so long as the number of draws grows appropriately, in any situation to which filtering methods can be applied, the Q-posterior can be used to produce Bayesian inferences that are well-calibrated. 

%Before concluding, we remark that the ease with which MCMC methods can now be implemented makes the Q-posterior attractive to those wishing to conduct purely frequentist inference: due to Fisher's identity, the required gradient estimates are available in closed-form, no second derivatives are requires to be calculated, and, critically, optimization of a noisy objective function is not required in order to produce consistent estimates of standard errors, and asymptotically valid inferences. 

\subsubsection{Variance Estimation}\label{sec:varest}
In order for Q-posteriors to be well-calibrated, a consistent estimator of $W_\star=\lim_{n}\text{Cov}\left\{m_n(\theta_\star)/\sqrt{n}\right\}$ is required. In cases where $m_n(\theta)$  takes the form $m_n(\theta)=\sum_{i=1}^{n}m_i(\theta)$, for some known functions $m_i(\theta)$ depending only on the $i$-th sample unit, the matrix $W_n(\theta)$ can  be taken as
\begin{equation}
{W}_n(\theta)=n^{-1}\sum_{i=1}^{n}\left\{m_i(\theta)-n^{-1}{m}_n(\theta)\right\}\left\{m_i(\theta)-n^{-1}{m}_n(\theta)\right\}^\top.
\label{eq:sampvar}	
\end{equation}
If one is worried that $W_n(\theta)$ in \eqref{eq:sampvar} may not be consistent for $W_\star$, we suggest the user employ bootstrapping techniques to estimate the covariance matrix; this can often be achieved using some variation of the estimating equations bootstrap (\citealp{hu2000estimating} and \citealp{chatterjee2005generalized}). Such an approach can either be directly embedded within the MCMC algorithm, or a preliminary estimator $W_n(\bar\theta_n)$, based on some $\bar\theta_n$ in the high-probability density region of the posterior,  can be used throughout the MCMC algorithm.\footnote{The latter follows since Assumption \ref{ass:weight} is actually a stronger condition than is required for the validity of Lemma \ref{lem:two}. The result can be shown to be satisfied if the matrix  $W_n(\theta)$ is replaced by $W_n(\hat\theta_n)$, where $\hat\theta_n$ is any preliminary consistent estimator of $\theta_\star$.}

For the Q-posterior $\overline\pi^Q_n(\theta)$, choosing $\widehat{W}_n(\theta;z)$ as the sample variance in \eqref{eq:sampvar} may not deliver a Q-posterior that is well-calibrated in general. This issue is readily seen even where $p_\theta(\y,\a)=\prod_{i=1}^{n}g_\t(y_i|\alpha_i)f_\t(\alpha_i)$, so that $\widehat{m}_n(\theta;z)$ takes the form $$\widehat{m}_n(\theta;z)=\sum_{i=1}^{n}N^{-1}\sum_{j=1}^{N}\nabla_\theta\log g_\t(y_i|\alpha_i^{(j)})f_\t(\alpha_i^{(j)}),\quad\a^{(j)}\sim p_\t(\a\mid\y),$$ which we note has variability in both $y_i$ and $\alpha_i^{(j)}$.  To see that both pieces of variability matter, write $m_{ij}(\theta)=\nabla_\theta \log g_\t(y_i\mid \alpha_i^{(j)})f_\t(\alpha_i^{(j)})$,
$\overline{m}_{iN}(\theta)=\frac{1}{N}\sum_{j=1}^{N}m_{ij}(\theta)$,  $z_i=(\alpha_i^{(1)},\dots,\alpha_i^{(N)})$, and consider
\begin{flalign*}
\Var\{\overline{m}_{iN}(\theta)\}&=\E_{Y}\left[\Var_{Z|Y}\{\overline{m}_{iN}(\theta)|Y\}\right]+\Var_{Y}\left[\E_{Z|Y}\{\overline{m}_{iN}(\theta)|Y\}\right].%\\&=N^{-1}\E_{Y}\left[\Var_{X^j_i|Y}\{\overline{m}_{i,j}(\theta)|Y\}\right]+\Var_{Y}\left[\E_{X^j_i|Y}\{m_{i,j}(\theta)|Y\}\right]
\end{flalign*}
Since $\a\stackrel{iid}{\sim}p_\t(\a\mid\y)$, $\E_{Z|Y}\{\overline{m}_{iN}(\theta)|Y\}=\E_{Z|Y}\{m_{ij}(\theta)|Y\}$, and this expectation can be estimated by $\overline{m}_{iN}(\theta)=\frac{1}{N}\sum_{j=1}^{N}m_{ij}(\theta)$. However,  applying the usual sample covariance estimator to $\overline{m}_{iN}(\theta)$ yields 
$$
\widehat{W}_{1n}(\theta;z):=n^{-1}\sum_{i=1}^{n}\left\{\overline{m}_{iN}(\theta)-n^{-1}\widehat{m}_n(\theta;z)\right\}\left\{\overline{m}_{iN}(\theta;z_i)-n^{-1}\widehat{m}_n(\theta;z)\right\}^\top.
$$Intuitively, $\widehat{W}_{1n}(\theta;z)$ is only an estimator of $\Var_{Y}\left[\E_{Z|Y}\{m_{i,j}(\theta)|Y\}\right]$ as the influence of the simulated $\a^{(j)}\sim p_\t(\a\mid\y)$ has already been `integrated out' via the averaging in $\overline{m}_{iN}(\theta)$. Consequently, unless $\Var_{Z|Y}\{\overline{m}_{iN}(\theta)|Y\}=0$ for all $i=1,\dots,N$, the usual sample covariance estimator for $\widehat{m}_n(\theta;z)$ under-estimates the variance, and delivers credible sets that are too narrow.  

However, given that $\widehat{W}_{1n}(\theta;z)$ is an estimator of $\Var_{Y}\left[\E_{Z|Y}\{\overline{m}_{iN}(\theta)|Y\}\right]$, all that remains is to obtain an estimator of $\E_{Y}\left[\Var_{Z|Y}\{\overline{m}_{iN}(\theta)|Y\}\right]$. Since $\a^{(j)}\stackrel{iid}{\sim} p_\t(\a\mid\y)$, an obvious choice is 
$$
\widehat{W}_{2n}(\theta;z)=\frac{1}{n}\sum_{i=1}^{n}\left[\frac{1}{N}\sum_{j=1}^{N}\{m_{ij}(\theta)-\overline{m}_{iN}(\theta)\}\{m_{ij}(\theta)-\overline{m}_{iN}(\theta)\}^\top\right],
$$where we note that the term in brackets is an estimator of $\Var_{Z|Y}\{m_{ij}(\theta)|Y\}$. Critically, since $\a^{(j)}\stackrel{iid}{\sim} p_\t(\a|\y)$ ($j=1,\dots,N$),    $\widehat{W}_{2n}(\theta;z)$ is consistent by the usual law of large numbers. This latter feature remains true even if the series $\{\overline{m}_{iN}(\theta):i\le n\}$ displays dependence. 

However, we do note that if $\{\overline{m}_{iN}(\theta):i\le n\}$ is dependent, an alternative estimator for  $\widehat{W}_{1n}(\theta;z)$ is required. Such an estimator can be achieved by replacing $\widehat{W}_{1n}(\t;z)$ with some heteroskedastic and auto-correlation (HAC) consistent covariance estimator (\citealp{newey1986simple}), or by using bootstrapping methods for the series  $\{\overline{m}_{iN}(\theta):i\le n\}$. For instance, in the case of dependent data, a version of the wild bootstrap (\citealp{shao2010dependent}) or the stationary bootstrap of \cite{politis1994stationary}, applied to the `data' $\{\overline{m}_{iN}(\theta):i\le n\}$, can be used to obtain consistent estimators of the variance.\footnote{Alternatively, the matrix $W(\theta)$ can be estimated using the waste-free SMC sampler \cite{dau2020waste} which allows us to obtain, from a single run of the algorithm, a consistent estimator of the asymptotic variance of our particle estimate $\widehat{m}_n(\theta;z)$.}

\subsection{Example: Random Effects Models}
\subsubsection{Linear Random Effects Models}\label{ex:regRE}

	To demonstrate the benefits of our generalized approach to Bayesian inference with latent variables we revisit the simple linear regression model in Example \ref{ex:reg}, and demonstrate that the performance of our approach is similar in cases with latent variables. To this  end, we consider the standard linear regression model with random effects
	$$
	y_{i}=\alpha_i+x_i^\top\beta+\sigma\epsilon_i,\;(i=1,\dots,n),
	$$where we assume that the error distribution $\epsilon_i\stackrel{iid}{\sim} N(0,1)$, and where $x_i$ and $\beta$ are $d_\beta$-dimensional vectors, and $\alpha_i$ is an unobservable random effect parameter with assumed distribution $\mathcal{N}(0,\sigma^2_\alpha)$. Similarly to the linear regression example given earlier, our main inferential goal is inference on the fixed effect parameters $\beta$.
	
Priors for $\beta$ and $\sigma$ are the same given before, and now we augment this with the additional prior for $\alpha_i|\sigma^2_\alpha\stackrel{iid}{\sim}N(0,\sigma^2_\alpha)$, where $\pi(\sigma^2_\alpha)\propto (\sigma^2_\alpha)^{-1}$. Under this prior choice, and conditional on $\theta=(\beta^\top,\sigma^2)^\top$,  and $\sigma^2_\alpha$, the conditional posterior for the random effects is known in closed-form: letting $\x$ denote the $n\times \mathrm{dim}(\beta)$-dimensional matrix with columns $x_{1,1:n}=(x_{1,i},\dots,x_{1,n})^\top$, and $I_{n}$ an $n\times n$-dimensional identity matrix, 
	$$
	p(\alpha_{1:n}\mid\y,\x,\beta,\sigma,\sigma^2_\alpha) 
		=N\left[\a;\left\{I_n\left(1+\frac{\sigma_{}^2}{\sigma^2_\alpha}\right) \right\}^{-1} (\y-\x {\beta}), \sigma_{}^2\left\{I_n\left(1+\frac{\sigma_{}^2}{\sigma^2_\alpha}\right) \right\}^{-1}\right].
	$$More generally, the conditional posteriors for $\beta$ and $\sigma^2$ are also known in closed-form. 
	
	Similarly to the linear regression example, suppose we again wish to ensure that our inferences for $\beta$ are robust to possible deviations from the modelling assumptions regarding the error term $\epsilon_i$; and in particular the possible form and presence of heteroskedasticity. Given that the conditional posterior $p(\alpha_{1:n}\mid\y,\x,\beta,\sigma,\sigma^2_\alpha)$ is known in closed-form  in this example, this can be achieved by estimating $m_n(\theta)$ and $W_n(\theta)$ using $m$ independent sets of draws from the conditional posterior $p(\alpha_{1:n}\mid\y,\x,\beta,\sigma,\sigma^2_\alpha)$. 
	
	Consequently, inference for $\theta=(\beta^\top,\sigma^2)^\top$ can proceed using a (pseudo-marginal) Metropolis-within-Gibbs scheme, whereby for each Metropolis step for $\theta$ we evaluate the loss by first sampling from $p(\alpha_{1:n}\mid\y,\x,\beta,\sigma,\sigma^2_\alpha)$, then forming $\widehat{Q}_n(\theta;z)$ and then deciding whether to accept or reject according to the usual criterion. We now implement such a sampling approach using $m=5$ simulations from the conditional posterior to estimate the gradient and its variance, at each step of the algorithm.

	Similarly to Example \ref{ex:reg}, we generate observed data from the linear regression model under the following specification for the disturbance term:
	$
	\epsilon_i\sim N(0,\sqrt{1+|x_{1,i}|^\gamma}^2),
	$ where $\gamma\ge0$, and $x_{1,i}$ denotes the first element of the vector $x_i$. Clearly, when $\gamma=0$ the DGP is homoskedastic, whereas if $\gamma\ne0$, the exact posterior will likely produce credible sets that are too narrow to yield reliable coverage. Alternatively, from Lemmas \ref{lem:two}-\ref{lem:three} suggest that the generalized posterior we propose correctly quantifies uncertainty regardless of the value of $\gamma$. 
	
	To this end, we generate $n=100$ observations from the assumed DGP under the designs of $\gamma=0$, and $\gamma=2$, where the $x_i$ are generated as tri-variate independent standard Gaussian random vectors, and we set the true values  as $\beta =(0.5,1.5,1,1)^\top$, and $\sigma =1$. We replicate this design 1000 times, and for each dataset sample the exact posterior using Gibbs sampling, and the generalized posterior using the RW sample in Algorithm \ref{alg1}, both posteriors are approximated using 10000 samples with the first 5000 discarded for burn-in. Across each data set and method, we then compare the posterior bias, variance and coverage for the regression coefficients. Table \ref{tab:regRE} summarises the repeated sampling results.
	\begin{table}[H]
		\centering
		{\footnotesize
			\begin{tabular}{lrrrrrr}
				\hline\hline
				\textbf{$\gamma=0$} & \multicolumn{3}{c}{Q-posterior} & \multicolumn{3}{c}{Exact} \\
				\textbf{} & Bias & Var &Cov& Bias & Var &Cov \\ \hline
$\beta_1$ & -0.0519 &   0.0310 &   0.9350 &  -0.0804 &   0.0555  &  0.9820\\
$\beta_2$ &-0.0017 &   0.0220 &   0.9320  & -0.0015 &   0.0154  &  0.8830\\
$\beta_3$ &-0.0132 &   0.0299 &   0.9450 &  -0.0136 &   0.0208  &  0.8880\\
$\beta_4$ &0.0079  &  0.0295  &  0.9200  &  0.0076  &  0.0191  &  0.8490\\\hline\hline
				\textbf{$\gamma=2$}& \multicolumn{3}{c}{Q-posterior} & \multicolumn{3}{c}{Exact}   \\ 
				\textbf{} & Bias & Var & Cov& Bias & Var &Cov \\ \hline
$\beta_1$ &   -0.0572 &   0.0308 &   0.9300 &  -0.0844 &   0.0536 &   0.9800\\
$\beta_2$ &-0.0047 &   0.0474 &   0.9150 &  -0.0060 &   0.0206 &   0.7620\\
$\beta_3$ &-0.0034 &   0.0312 &   0.9290 &  -0.0034 &   0.0194 &   0.8740\\
$\beta_4$ &0.0000  &  0.0296 &   0.9290  &  0.0001 &   0.0191  &  0.8480\\\hline\hline
		\end{tabular}}
		\caption{Posterior accuracy results in the normal linear regression model for the exact and Q-posterior.  Bias is the bias of the posterior mean across the replications. Var is the average posterior variance deviation across the replications. Cov is the posterior coverage. The nominal level is set to be 95\% for the experiments. }
		\label{tab:regRE}
	\end{table}

	Similarly to the simple linear regression example, the results in Table \ref{tab:regRE} demonstrate that our generalized Bayes approach based on estimated likelihood gradients produces results that are similar to exact Bayes when the model is correctly specified. However, while the coverage rates for exact Bayes are relatively far from the nominal levels under both correct specification $(\gamma=0)$ and misspecification $(\gamma=1)$, our approach produces coverage rates that are very close to the nominal 95\% level in both settings.

\subsubsection{Generalized Linear Random Effects Models}\label{sec:randomE}
We observe binary outcomes on individual $i$ from group $j$, $y_{ij}$, and let $x_{ij}$ denote a vector of strictly exogenous regressors. In particular, we believe that the outcomes are generated according to a binary random effects probit model of the form
\begin{equation}
	\label{probitmodel}
	y_{ij}\sim \mathsf{Bernoulli}(p_{ij}),\quad p_{ij}=\Phi\left(x_{ij}^\top\beta+\alpha_i\right),
\end{equation}
for individual $i=1,\dots,n$, with $j=1,\dots,J$ observations per individual, where $\alpha_i$ denotes the unobservable random effect for individual $i$, and $\Phi(x)$ denotes the standard Gaussian CDF. Similarly to the linear random effects case, we again take the prior on $\beta$ to be diffuse, $\pi(\beta)\propto1$. The standard assumption on the random effects parameter is that $\alpha_i\stackrel{iid}{\sim}N(0,\sigma^2_\alpha)$. However, the very nature of the random effect term makes its precise parametric representation unclear. 

This section  demonstrates how to apply the Q-posterior in this model, and compares our approach against the exact posterior in the case where the random effects are correctly and incorrectly specified. For simplicity, we take $J=1$, and denote quantities using only the $i$ sub-script, however, the case $J>1$ proceeds similarly.

Conditional on $\theta$, the complete-data likelihood is  
$$
 p\left(y_{1:n},\alpha_{1:n}\mid\theta\right)=\prod_{i=1}^{n}\left\{\Phi\left(x_{i}^\top\beta+\alpha_{i}\right)
 \right\}^{y_{i}}\left\{1-\Phi\left(x_{i}^\top\beta+\alpha_{i}\right)\right\}^{1-y_i}N(\alpha_i;0,\sigma^2_\alpha);
$$note that the only occurrence of $\sigma^2_\alpha$ comes from the contribution in the density components $N(\alpha_i;0,\sigma^2_\alpha)$. Using Fisher's identity, the log gradient with respect to $\beta$  is  
\begin{equation*}
	\nabla_\beta \log p(\y,\alpha_{1:n}\mid\theta)=\sum_{i=1}^{n} x_{i}^\top \int_{\mathbb{R}}u_i(\theta,\alpha_i)p_\t(\alpha_i\mid y_{1:n})\dt\alpha_i, 
\end{equation*}where 
$$
 u_i(\theta,\alpha_i)=\frac{\left[y_{i}-\Phi\left(x_{i}^\top \beta+\alpha_i\right)\right]\varphi\left(x_{i}^\top \beta+\alpha_{i}\right)}{\Phi\left(x_{i}^\top \beta+\alpha_i\right)\left[1-\Phi\left(x_{i}^\top \beta+\alpha_i\right)\right]},
$$
while the gradient with respect to $\sigma^2_\alpha$ is
$$
\nabla_{\sigma^2_\alpha} \log p_\theta(\y,\alpha_{1:n})=-\frac{n}{2\sigma^2_\alpha}+\frac{1}{2}\left(\frac{1}{\sigma^2_\alpha}\right)^2\sum_{i=1}^{n} \int_{\mathbb{R}}\alpha_i^2p_\t(\alpha_i\mid y_{1:n})\dt\alpha_i.
$$ Unlike the linear case, $p_\t(\alpha_{1:n}\mid\y)$ is analytically intractable in this generalized linear model. However, following the approaches of \cite{zeger1991generalized} and \cite{mcculloch1997maximum}, it is possible to reliably generate samples from $p_\t(\alpha_{1:n}\mid\y)$ using a Metropolis-Hastings algorithm based on the proposal $N(0,\sigma^2_\alpha)$. 

It is then feasible to use a Metropolis-within-Gibbs approach to evaluate the Q-posterior, whereby at the $k$-th step we sample $N$ paths $\alpha^{}_{1:n,j}(k)\sim p_\t(\alpha_{1:n}\mid\y)$, $j=1,\dots,N$, and then use these draws to estimate the gradients.
%$$
%\widehat{m}_n(\theta;{z})=N^{-1}\sum_{j=1}^{N}\sum_{i=1}^{n}\begin{bmatrix}x^\top_{i}u_i\{\theta,\alpha_{i,j}^{}(k)\},&\sigma^2_\alpha-{\alpha_{i,j}^{2}(k)}
%\end{bmatrix}^\top.
%$$
Once the gradients have been estimated, their variance can be estimated using the methods discussed in Section \ref{sec:varest}. 

We now compare the behavior of the Q-posterior and the exact posterior in the random effects probit model. We generate $n=100$ sample observations from the random effects probit model under two different distributions for the random effects distribution: the first DGP correctly assumes that the random effect distribution is $N(0,\sigma^2_\alpha)$; while in the second  DGP the actual distribution of the random effects is student-t with four degrees of freedom, so that the assumed model is misspecified under the second design. The regressor $x_i$ is again generated as tri-variate independent standard Gaussian random vectors, and we set the true values  as $\beta =(1,1,1,1)^\top$. In both cases the scale of the random effects parameter is set to unity.

We replicate these two design 1000 times, and for each dataset sample the exact posteriors using a Gibbs-within-Metropolis sampling scheme, similar to \cite{zeger1991generalized}. The Q-posterior is also sampled using a Gibbs-within-Metropolis discussed above, and where we use $N=5$ draws from the conditional posterior in all experiments.\footnote{Since this is a Gibbs-within-Metropolis scheme, we run the sampler to obtain $2*N$ draws of $\alpha_i$, and then discard the first $N$ for burn-in.} Both posteriors are approximated using 10000 samples with the first 5000 discarded for burn-in. Across each data set and method, we then compare the posterior bias, variance and coverage for the regression coefficients. Table \ref{tab:GregRE} summarises the repeated sampling results.

Similarly to the other examples, the Q posterior delivers posterior locations that are similar to those of exact Bayes, but again its uncertainty quantification remains close to the nominal level regardless of correct or incorrect model specification. In contrast, since the variances of exact Bayes are (on average) smaller than those given by the Q-posterior, the coverage of the exact Bayes posteriors is generally further away from the nominal level than for the Q-posterior. This difference is particularly stark for the corresponding posteriors themselves, which we plot in Figure \ref{fig:two} for one particular replication under the misspecified DGP. In this case, it is clear that while the posteriors have similar locations, their scales are quite different. It is this increase in uncertainty that produces reliable uncertainty quantification even though the model is misspecified.

	\begin{table}[H]
	\centering
	{\footnotesize
		\begin{tabular}{lrrrrrr}
			\hline\hline
			 & \multicolumn{3}{c}{Q-posterior} & \multicolumn{3}{c}{Exact} \\
			\textbf{$N(0,1)$} & Bias & Var &Cov& Bias & Var &Cov \\ \hline
			$\beta_1$ & -0.0438  &  0.0358  &  0.9760  & -0.1016  &  0.0266 &   0.9060\\
			$\beta_2$ &-0.0414  &  0.0375 &   0.9740  & -0.0841 &   0.0306 &   0.9420\\
			$\beta_3$ &-0.0360  &  0.0339 &   0.9680 &  -0.0955  &  0.0344  &  0.9160\\
			$\beta_4$ &-0.0352  &  0.0396 &   0.9780 &  -0.0863 &   0.0298 &   0.9280\\\hline\hline
			& \multicolumn{3}{c}{Q-posterior} & \multicolumn{3}{c}{Exact}   \\ 
			\textbf{$t_4$} & Bias & Var & Cov& Bias & Var &Cov \\ \hline
$\beta_1$& -0.1255  &  0.0365 &   0.9440 &  -0.1638 &   0.0249  &  0.8350\\
$\beta_2$&-0.1053  &  0.0401  &  0.9480 &  -0.1509  &  0.0273  &  0.8450\\
$\beta_3$&-0.1208  &  0.0391  &  0.9440 &  -0.1607  &  0.0277  &  0.8450\\
$\beta_4$ &-0.1171  &  0.0405  &  0.9490 &  -0.1595  &  0.0282 &   0.8520\\\hline\hline
	\end{tabular}}
	\caption{Posterior accuracy results in the probit random effects model for the exact and Q-posteriors.  Bias is the bias of the posterior mean across the replications. Var is the average posterior variance deviation across the replications. Cov is the posterior coverage (95\% nominal coverage).  The first design ($N(0,1)$) represents the case where the distribution of random effects is correctly specified, while the design $t_4$ represent the case of incorrect specification for the random effects distribution.}
	\label{tab:GregRE}
\end{table}

\begin{figure}[htb]
	\centering
	\includegraphics[width=14cm, height=8cm]{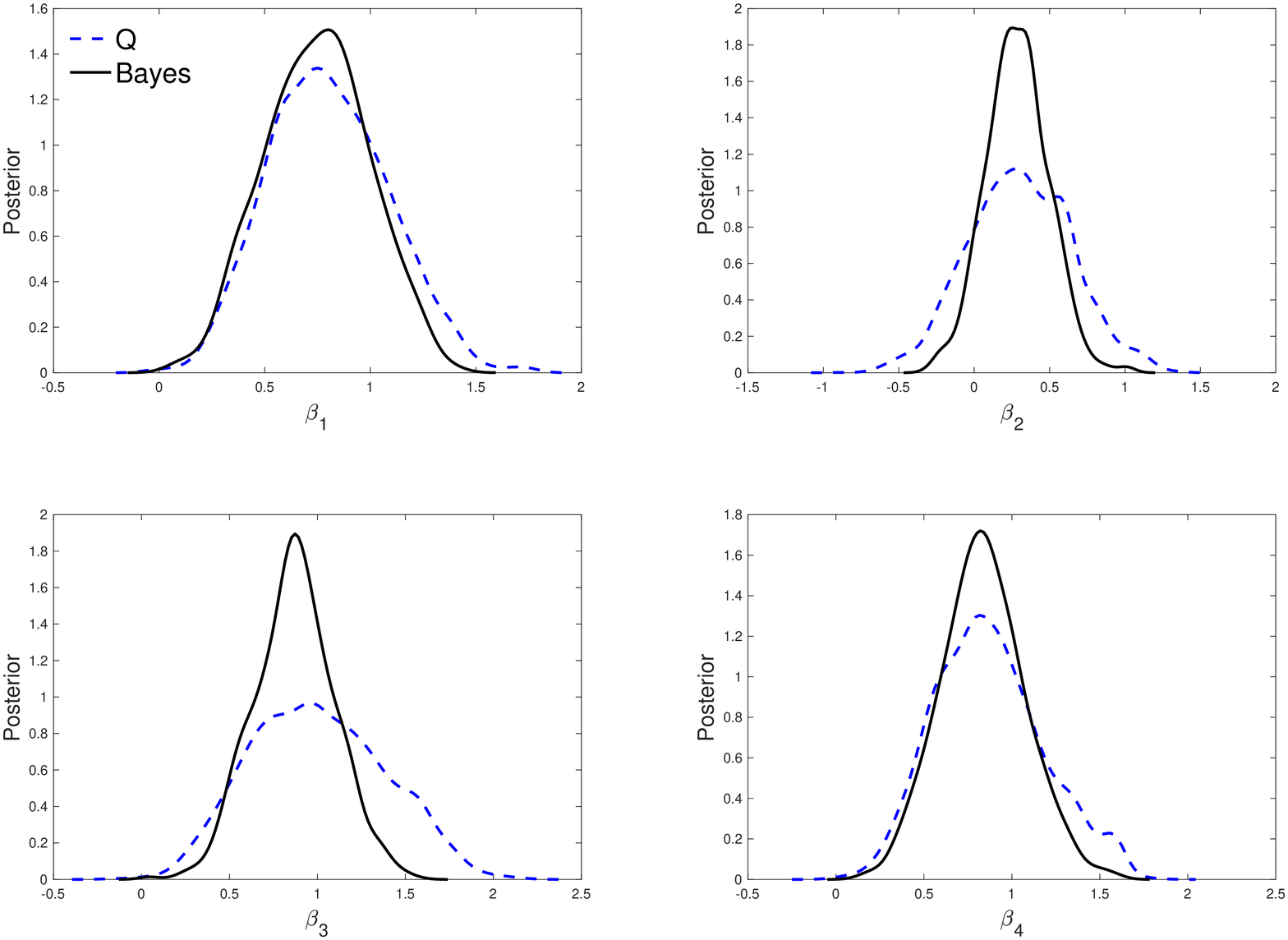}  
	\caption{Exact and Q-posterior comparison in a representative dataset in the random effects probit model when the distribution of random effects is misspecified. }
	\label{fig:two}
\end{figure}

\section{General Loss Function Posteriors}\label{sec:general}

When the model is misspecified there is no reason to assume that the KL minimizer $\theta_\star$ delivers inferences that are fit for purpose. A solution to this issue is to instead conduct inference using generalized posteriors, which produce Bayesian inferences using a loss function that is specific to the problem at hand. To ensure Bayesian inferences remain fit for purpose in misspecified models, \cite{bissiri2016general} suggest replacing the likelihood posterior update with that based on a loss function $q:\mathcal{Y}\times\Theta\rightarrow\mathbb{R}_+$ that is important to the specific inferential task at hand. If one then considers $q_n(\theta)=\sum_{i=1}^{n}q(y_i,\theta)$ as a cumulative loss function over the observed sample, a posterior update can be based on the  generalized posterior
$$
\pi^q_n(\theta)=\frac{\exp\left\{-\omega q_n(\theta)\right\}\pi(\theta)}{\int_\Theta\exp\left\{-\omega q_n(\theta)\right\}\pi(\theta)},
$$for some learning rate $\omega\ge0$.

While useful, generalized posteriors are not well-calibrated, and do not correctly quantify uncertainty in general; see Section 4 of \cite{miller2021asymptotic} for a discussion, and the references therein, for a discussion of existing methods to circumvent this issue.  In this section, we demonstrate that the Q-posterior can be equally applied to generalized Bayesian posteriors to deliver valid uncertainty quantification. Given $q_n(\theta)$, we can represent the information in this loss by the set of $d_\theta$-dimensional estimating equations $\psi_n(\theta)=\nabla_\theta q_n(\theta)$, where we assume that these derivatives exists for all $\theta\in\Theta$, and each $i\ge1$.\footnote{In general, all we will require is that this derivative exists with probability one, which permits the case where the derivatives are not defined at a  countable collection of points; e.g., such as the case where our loss function is median loss.} Letting $V_n(\theta)$ be an estimator of $\text{Cov}\{\psi_n(\theta)/\sqrt{n}\}$, we can update our prior beliefs $\pi(\theta)$ to posterior beliefs using the following Q-posterior
\begin{equation}\label{eq:gengenpost}
	\pi^\Psi_n(\theta)=\frac{\exp\{-\Psi_n(\theta)\}\pi(\theta)}{\int_\Theta\exp\{-\Psi_n(\theta)\}\pi(\theta)\dt\theta},\quad\Psi_n(\theta):=\frac{1}{2}\frac{{\psi_n(\theta)^\top}}{\sqrt{n}} V_n(\theta)^{-1}\frac{\psi_n(\theta)}{\sqrt{n}}.	
\end{equation}

Given the results in Section \ref{sec:large}, the Q-posterior $\pi^\Psi_n(\theta)$ should deliver a generalized Bayesian posterior that correctly quantifies uncertainty without the need to choose any tuning constants; see Theorem \ref{lem:qtwo} for details. Before presenting this result, we first demonstrate the applicability of the Q-posterior $\pi^\Psi_n(\theta)$ by showing in a simple example  that $\pi^\Psi_n(\theta)$ correctly quantifies uncertainty in the context of Bayesian inference for an unknown parameter based on a general loss function.
\subsection{Robust Quantile Inference}\label{ex:med}
	Consider observing an iid sequence $y_1,\dots,y_n$ from the model 
	$$
	y_i=\theta+\epsilon_i,\quad i=1,\dots,u,\quad \epsilon_i\stackrel{iid}{\sim}F(\cdot)
	$$where the unknown parameter $\theta$ has prior density $\pi(\theta)$ and is independent of $\epsilon$. Following \cite{doksum1990consistent}, for reasons  of robustness, we do not wish to estimate the posterior of $\theta$ from $\{y_{i}:i\ge1\}$ but from the sample median $T_n=\text{med}(y_1,\dots,y_n)$; for a related approach see \cite{lewis2021bayesian}. When $F$ has density $f$, this then gives rise to the exact posterior distribution 
	$$
	\pi(\theta|T_n)  = \pi(\theta)\exp\left\{\frac{1}{2}(n-1)\log F(T_n-\theta)\{1-F(T_n-\theta)\}+\log f_\theta(T_n)(1-F))\right\}
	$$
	when $n$ is odd.
	
	Given the above form of the posterior, \cite{miller2021asymptotic} suggests instead generalized Bayesian inference for $\theta$ using the simpler loss function 
	$
	-\frac{1}{2}\log F(T_n-\theta)\{1-F(T_n-\theta)\}
	$. While \cite{miller2021asymptotic} demonstrates that such a posterior is well-behaved in large samples, the resulting posterior does not have correct coverage even when the model is correctly specified. In contrast, we suggest conducting Bayesian inference using the Q-posterior, which would be based on the gradient of the above loss, with respect to $\theta$,
	$$
	m_n(\theta)=\frac{1}{2}\frac{f(T_n-\theta)}{F(T_n-\theta)}-\frac{1}{2}\frac{f(T_n-\theta)}{1-F(T_n-\theta)}.
	$$The variance of the $m_n(\theta)$, for any $\theta\in\Theta$, can be consistently estimated using the iid bootstrap. Letting $W_n(\theta)$ denote such an estimate; we can then conduct inference on $\theta$ using the Q-posterior 
	$$
	\pi_n(\theta)\propto \pi(\theta) \exp\left[-\frac{1}{2}\left\{\frac{1}{2}\frac{f(T_n-\theta)}{F(T_n-\theta)}-\frac{1}{2}\frac{f(T_n-\theta)}{1-F(T_n-\theta)}\right\}^2 W_n(\theta)^{-1}\right].
	$$In this example, we use the estimating equations bootstrap (see, \citealp{hu2000estimating} and \citealp{chatterjee2005generalized} for discussion) to compute the matrix $W_n(\theta)$ at every value of $\theta$.

	We now compare the uncertainty quantification produced using the generalized and Q-posteriors in correctly and misspecified models. In both cases, we assume $F(\cdot)$ is a standard Gaussian distribution. In the first experiment, referred to as DGP1, observed data is generated from a Gaussian distribution with mean $\theta=1$, and with variance 4. In the misspecified regime, referred to as DGP2, we generate observed data from a mixed Gaussian distribution with parameterization
	$
	0.9N(\theta=1,4)+0.1N(0,1).
	$ 
	
	In the first case the true median is unity, while in the second case the true median of the data is actually approximately $0.84$.\footnote{This value must be found numerically by inverting the cdf of the mixture distribution.} We simulate 1000 observed datasets from both DGPs and compare the results of generalized Bayes and that based on the Q-posterior. Table \ref{tab:med} compares the posterior means, variances and coverage of the two procedures. The results demonstrate that the generalized posterior suggested by \cite{miller2021asymptotic} does not produce reliable coverage for the true median, while the coverage of the Q-posterior is again close to the nominal level. 
	
	\begin{table}[H]
		\centering
		{\footnotesize
			\begin{tabular}{lrrrrrr}
				\hline\hline
				& \multicolumn{3}{c}{Q-posterior} & \multicolumn{3}{c}{Exact} \\
				\textbf{DGP1}  & Bias & Var &Cov& Bias & Var &Cov \\ \hline
				$\theta$ &   -0.0031  &  0.0654 &   0.9420  & -0.0031  &  0.0152   & 0.6700\\
				%& \multicolumn{3}{c}{Q-posterior} & \multicolumn{3}{c}{Exact}   \\ 
				\textbf{DGP2}  & Bias & Var & Cov& Bias & Var &Cov \\ \hline
				$\theta$ &     -0.0115  &  0.0607  &  0.9360  & -0.0118  &  0.0157 &   0.7160
				\\
				\hline\hline
		\end{tabular}}
		\caption{Posterior accuracy results for the median using the generalized posterior and Q-posterior.  Bias is the bias of the posterior mean across the replications. Var is the average posterior variance deviation across the replications. Cov is the posterior coverage (95\% nominal coverage). }
		\label{tab:med}
	\end{table}

\subsection{Uncertainty Quantification}
If the regularity conditions maintained for $m_n(\theta)$ and $W_n(\theta)$ in the likelihood-based Bayesian context are satisfied for $\psi_n(\theta)$ and $V_n(\theta)$ in the generalized Bayesian context, then the posterior $\pi^\Psi_n(\theta)$ correctly quantifies uncertainty.

\begin{theorem}\label{lem:qtwo}Consider that Assumption \ref{ass:prior} is satisfied and that there exist  non-random functions $\psi:\Theta\rightarrow\mathbb{R}^{d_\t}$ and $V:\Theta\rightarrow\mathbb{R}^{d_\t\times d_\t}$ such that Assumptions \ref{ass:infeasible} and \ref{ass:weight} are satisfied with $m_n(\theta)=\psi_n(\theta)$, $m(\theta)=\psi(\t)$, and $W_n(\theta)=V_n(\theta)$. Then, for $\vartheta:=\sqrt{n}(\theta-\theta_\star)-Z_n/\sqrt{n}$, $Z_n:=\Omega_\star^{-1}\nabla_{\theta}\psi(\theta_\star)^\top V(\theta_\star)^{-1}\psi_n(\theta)$, $\Omega_\star=\left[\nabla_{\theta}\psi(\theta_\star )^\top V(\theta_\star)^{-1}\nabla_{\theta}\psi(\theta_\star )\right]$, and $\mathcal{T}_n:=\{\vartheta=\sqrt{n}(\theta-\theta_\star)-Z_n/\sqrt{n}:\theta\in\Theta\}$,  as $n\rightarrow+\infty$, 
	$
	\int_{\mathcal{T}_n}\|\vartheta\||\pi^{\Psi}_{n}(\vartheta)-N\{\vartheta;0,\Omega_\star^{-1}\}|\dt 
	\vartheta=o_p(1),
	$
 and, for $\overline\theta=\int_{\theta}\theta\pi^\Psi_n(\theta)\dt\theta$, 
	$
	\sqrt{n}(\overline\theta-\theta_\star)\Rightarrow N(0,\Omega_\star^{-1}).
	$
\end{theorem}

The generalized Bayesian posterior of \cite{bissiri2016general} has credible sets whose width is determined by $[\nabla_{\theta}\psi(\theta_\star)]^{-1}$, whereas, those of the Q-posterior $\pi^\Psi_n(\theta)$ are determined by $\Omega(\theta_\star)^{-1}=[\nabla_{\theta}\psi(\theta_\star)]^{-1} V(\theta_\star)^{}[\nabla_{\theta}\psi(\theta_\star)]^{-1}$, which ensures that $\pi^\Psi_n(\t)$ correctly quantifies uncertainty. Existing approaches to coverage correction in generalized Bayesian inference rely on ex-post processing of the posterior and calculating second-derivative information, e.g., \cite{holmes2017assigning}, \cite{syring2019calibrating}, \cite{giummole2019objective}, or the application of expensive bootstrapping procedures, e.g., \cite{matsubara2021robust}. In contrast,   our approach requires no post-processing and delivers an intrinsically (generalized) Bayesian solution to the problem of well-calibrated generalized posterior inferences.

\begin{remark}\normalfont 
	A common class of generalized posteriors occurs when $q(y_i,\theta)$ is itself obtained from a pseudo-likelihood of some sort, including partial likelihoods (\citealp{cox1975partial}) and composite likelihoods (\citealp{lindsay1988composite}, and \citealp{varin2011overview}). For example, Sections 5 and 7 of \cite{miller2021asymptotic} give several useful examples of such approaches and their link with generalized posteriors. However, as acknowledged by \cite{miller2021asymptotic}, generalized posteriors built from these loss functions ``do not exhibit correct coverage, even asymptotically.'' However, in each of the cases covered in \cite{miller2021asymptotic}, the loss function is sufficiently smooth to permit its representation in terms of the Q-posterior $\pi^\Psi_n(\t)$ in \eqref{eq:gengenpost}. Hence, for instance, so long as the composite likelihood is sufficiently smooth in $\theta$, we can apply the Q-posterior to produce generalized Bayesian inferences based on this loss function that have correct coverage. 
\end{remark}

%\subsubsection{Large Sample Behavior}

\section{Discussion and Conclusions}\label{sec:discuss}
\subsection{ Alternative Approaches}\label{sec:discuss1}
In a likelihood context, when the model is correctly specified we have that $\Sigma_\star^{-1}=W_\star^{-1}=\mathcal{I}(\theta_\star)^{-1}$, and the Q-posterior and exact posterior asymptotically agree. However, if the model is misspecified, the Q-posterior still yields credible sets that are well-calibrated; such a result will only hold if $W_n(\theta_\star)$ is a consistent estimator of $W_\star$. Otherwise, the Q-posterior will not correctly quantify uncertainty. As discussed in Section \ref{sec:varest}, however, in most cases reliable estimators of $W(\theta_\star)$ are available using existing formulas, or bootstrapping methods.

The Q-posterior approach represents a significant departure from existing approaches to Bayesian inference in possibly misspecified models. Two approaches that have so far received meaningful attention are the `sandwich' correction suggested in \cite{muller2013risk}, and the BayesBag approach, see \cite{huggins2019robust}. 

\citeauthor{muller2013risk}'s approach amounts to correcting the draws from the standard posterior using the explicit Gaussian approximation $\theta\sim N\{\bar\theta,\mathcal{H}_n(\bar\theta)^{-1}W_n(\bar\theta)\mathcal{H}_n(\bar\theta)^{-1}\}$, where $\bar\theta$ is the posterior mean. Such a correction can be implemented either by drawing directly from  a multivariate normal, or by taking each posterior draw $\theta$ and modifying it according to the linear equation $$\widetilde\theta=\bar\theta+\mathcal{H}_n(\bar\theta)^{-1}W_n^{}(\bar\theta)\mathcal{H}_n(\bar\theta)^{1/2}(\theta-\bar\theta);$$see, also, \cite{giummole2019objective} for a related approach in the case of generalized posteriors built using scoring rules. We argue that this \textit{ex-post correction} is sub-optimal for several reasons: firstly, philosophically, it amounts to the application of frequentist principles to the output of a Bayesian learning algorithm, and thus is not intrinsically Bayesian; second, it requires the explicit calculation of second-derivatives, which can be difficult and may be ill-behaved; thirdly, this Gaussian approximation is poor when posteriors are not roughly Gaussian, such as when the parameters have restricted support, or when we have small sample sizes; finally, this correction can easily produce a value of $\tilde\theta$ lying outside the support of $\pi(\theta)$, for instance, when $\Theta$ is a bounded subset of $\mathbb{R}^{d_\theta}$.\footnote{While it may be feasible to transform the draws so that they are restricted to the appropriate space, such transformations may not be invariant, and the choice of which transformation to utilize generally has no theoretical basis.} 

An alternative approach is the use of posterior bagging, as suggested in the BayesBag approach; see \cite{huggins2019robust} for an in-depth discussion of BayesBag. This procedure  attempts to correct posterior coverage through  bagging. Letting $b=1,\dots,B$ denote bootstrap indices, and $\y^{(b)}=(y_1^{(b)},\dots,y_n^{(b)})$ the $b$-th bootstrap sample, where $y_i^{(b)}$ is sampled with replacement from the original dataset. The BayesBag posterior is given by 
$$
\pi^\star(\theta \mid \y)\approx B^{-1}\sum_{b=1}^{B}\pi^\star(\theta\mid \y^{(b)}).
$$BayesBag is easy to use, and requires no additional algorithmic tools. However, it does require re-running the MCMC sampling algorithm to obtain posterior draws of $\theta$ for each $\{\y^{(b)}:b=1,\dots,B\}$. The latter can be computationally intensive if the model is high-dimensional or contains latent variables. Even then, \cite{huggins2019robust} demonstrate that the BayesBag posterior has credible sets that are not well-calibrated. Asymptotically, the BayesBag posterior variance is given by 
$
c^{-1}\mathcal{H}_\star^{-1}+c^{-1}\Sigma^{-1}_\star,
$ where $c=\lim_n B/n$. Hence, depending on the choice of $B$, the BayesBag posterior displays over-or under-coverage. Only when the parameter $\t$ is scalar valued, is it possible to choose $c$ such that the posterior has valid frequentist coverage; i.e., such that  
$
\Sigma_\star^{-1}=c(\mathcal{H}^{-1}_\star+\Sigma_\star^{-1}).
$ Lastly, the applicability  of BayesBag approach itself to models with weakly dependent data, or when the likelihood is intractable,   does not seem straightforward.

In contrast with the ex-post correction approach and BayesBag, the Q-posterior does not require any ex-post sampling correction, or bootstrapping approaches to obtain well-calibrated posteriors. Moreover, it is feasible to use the Q-posterior in cases where pseudo-marginal methods are required to conduct Bayesian inference.

\subsection{Conclusion}\label{sec:conclusions}

We have proposed a new approach to Bayesian inference, referred to generally as Q-posteriors, that deliver reliable uncertainty quantification. In likelihood-based settings the Q-posterior can be thought of as a type of Bayesian synthetic likelihood (\citealp{wood2010statistical}) posterior where we replace the likelihood for the observed sample with an approximation for the likelihood of the score equations. The critical feature of the Q-posterior is that it is guaranteed to deliver credible sets that are asymptotically well-calibrated regardless of model specification, while if the model is correctly specified the Q-posterior agrees with the exact posterior (in large samples). Critically, even when the likelihood must be estimated due to the presence of latent variables, the Q-posterior still delivers well-calibrated inferences under fairly weak regularity conditions. %Interestingly, the reliable uncertainty quantification produced by the Q-posterior is achieved without the need to estimate second derivative information of the (estimated) likelihood, as would be the case with frequentist-based approaches for such models. 

When applied to generalized Bayesian posteriors (\citealp{bissiri2016general}), we have shown that the Q-posterior remains well-calibrated. All existing approaches of which we are aware attempt to correct the coverage of generalized posteriors use either ex-post correction of the posterior draws, which are ultimately based on some (implicit) normality assumption on the resulting posterior draws, or  may require complicated bootstrapping approaches. In contrast, the Q-posterior delivers correct uncertainty quantification without the need of any additional tuning or ex-post correction of the draws.

When the likelihood is intractable and must be estimated, a version of the Q-posterior can be constructed using Fisher's identity and simple Monte Carlo estimators. In principle, it would be feasible to apply other Monte Carlo approaches, such as importance sampling, to construct the Q-posterior. This paper restricts itself to simple Monte Carlo estimators as they are the easiest to apply, and theoretically analyze. In ongoing work by the authors, we explore the use of sequential importance sampling estimators, and the resulting performance of the Q-posterior is similar. However, in such cases, analysing the behavior of the posterior becomes more difficult due to the sequential nature of the Monte Carlo estimators, and obtaining theoretical results similar to those in the Theorem \ref{thm:main} becomes more onerous. We leave this important question for future research.

\singlespacing
\bibliographystyle{apalike}
\bibliography{library}

\appendix

\onehalfspacing

\section{Main Lemmas}\label{app:Lemmas}
This section contains  lemmas that are used to prove the results in the main text, and are of independent interest. We first establish the following additional notation used throughout the remainder of this appendix.  For two (possibly random) sequences $a_n,b_n$ we say that  $a_n\lesssim b_n$ if for some $n'$ large enough, and all $n\ge n'$, there exists a $C>0$ such that $a_n\le C b_n$ (almost surely); while we write  $a_n\asymp b_n$ if $a_n\lesssim b_n$ and $b_n\lesssim a_n$ (almost surely).  For a known function $f:\Theta\rightarrow\mathbb{R}^d$,  we will denote the expectation of $f$ wrt a probability measure $H$ as $\E_{H}[f(\theta)]=\int_\Theta f(x)\dt H(x)$.  Throughout, we abuse notation and let $\|\cdot\|$ denote the Euclidean norm in the case of vectors or a convenient matrix norm. The latter abuse of notation is immaterial since we will only treat vectors and matrices of fixed length.

\begin{lemma}\label{lem:four_sim}
	Let $Z\in\mathbb{R}^d$ with $\mu=\E[Z]$, and let $Z_l$ (resp., $\mu_l$), $l=1,\dots, d$, denote the coordinates of $Z$ (resp., $\mu$). Let $M$ be a positive semi-definite matrix with $M=A^\top A$ for some matrix $A$. Assume that,  for some $\sigma>0$, and each $l=1,\dots,d$, $|\mu_j|\le \sigma$, and for $j\in\{2,4\}$, $\E[|Z_l-\mu_l|^j]\le \sigma^j$, then 
	$
	\text{Var}(\|AZ\|^2)\le 2\sigma^4\|M\|^2(d+2)^2.
	$
\end{lemma}
\begin{proof}
	Write 
	\begin{flalign*}
		\|AZ\|^2=Z^{\top}MZ&=(Z-\mu)^{\top}M(Z-\mu)-\mu^{\top}M\mu+2\mu^{\top}MZ\\&=\|A(Z-\mu)\|^2-\|A\mu\|^2+2\mu^{\top}MZ.	
	\end{flalign*}
	Then,
	\begin{flalign*}
		\text{Var}(\|AZ\|^2)&=\text{Var}\{\|A(Z-\mu)\|^2\}+4\text{Var}\{\mu^{\top}MZ\}+2\text{Cov}\{2\mu^{\top}MZ,\|A(Z-\mu)\|^2\}.
	\end{flalign*}From Cauchy-Schwartz,
	$$
	|\text{Cov}\{\mu^{\top}MZ,\|A(Z-\mu)\|^2\}|\le \sqrt{\text{Var}\{\mu^{\top}MZ\}}\sqrt{\text{Var}\{\|A(Z-\mu)\|^2\}}
	$$so that
	\begin{flalign*}
		\text{Var}(\|AZ\|^2)&\le\text{Var}\{\|A(Z-\mu)\|^2\}+4\{\mu^{\top}M\Sigma M\mu\}+2\sqrt{4\{\mu^{\top}M\Sigma M\mu\}}\sqrt{\text{Var}\{\|A(Z-\mu)\|^2\}}.
	\end{flalign*}
	
	Now, we bound $\text{Var}[\|A(Z-\mu)\|^2]$ using the moment hypothesis in the statement of the result. In particular, 
	$$
	\text{Var}[\|A(Z-\mu)\|^2]\le \|M\|^2\text{Var}[\|Z-\mu\|^2],
	$$ and
	\begin{flalign*}
		\text{Var}[\|(Z-\mu)\|^2]&=\sum_{i=1}^{d}\text{Var}[(Z_i-\mu_i)^2]+2\sum_{i=1}^{d}\sum_{j, i\ne j}^d\text{Cov}\left[(Z_i-\mu_i)^2,(Z_i-\mu_j)^2\right]\\&\le d\sigma^4+2\sum_{i=1}^{d}\sum_{j, i\ne j}^d\text{Cov}\left[(Z_i-\mu_i)^2,(Z_i-\mu_j)^2\right]\\&\le d\sigma^4+2\sum_{i=1}^{d}\sum_{j, i\ne j}^d\left\{\text{Var}[(Z_i-\mu_i)^2]\text{Var}[(Z_j-\mu_j)^2]\right\}^{1/2}\\&=d\sigma^4+2d(d-1)\sigma^4\\&\le2\sigma^4d^2.
	\end{flalign*}

For	$A,B$ positive semi-definite matrices of conformable dimension,  we have the following inequality
$$
\tr AB \le \|A\| \tr (B).
$$ 
Using properties of $\tr(\cdot)$, and the above inequality, we can obtain
	\begin{flalign*}
		\mu^{\top}M\Sigma M\mu=\text{tr}(\mu^{\top}M\Sigma M\mu)\le& \sigma^2\tr(M\mu\mu^\top M)\\\le & \sigma^2\|M\|^2\|\mu\|^2\\\le&
		\sigma^4\|M\|^2.
	\end{flalign*}%since $\|\mu\|^2\le C\sigma^2$ for some $C>0$ (by definition of $\sigma^2$).
	Applying the two displayed equations then yields
	$$
	\text{Var}(Z^{\top}MZ) \le 2\sigma^4\|M\|^2[d^2+2+2d]=\sigma^2\|M\|^2(d+2)^2.
	$$
\end{proof}
The following lemma is used to extend the result in Lemma \ref{lem:four_sim} to the case where the matrix in the quadratic form is random.  
\begin{lemma}\label{lem:matlem}Let $A$ be positive-definite, and $B=A+H$ with $\rho:=\|A^{-1}H\|<1$, then 
	$$
	\|B^{-1}-A^{-1}\|\le \|A^{-1}\|^{2}\|H\|/(1-\rho).
	$$	
\end{lemma}
\begin{proof}Since $B=A(I+A^{-1}H)$, and $\rho<1$, $B$ is non-singular. Write
	$$
	A^{-1}-B^{-1}=A^{-1}(B-A)B^{-1}=A^{-1}HB^{-1},
	$$ so that 
	\begin{flalign}\label{eq:matbound} 
	\|A^{-1}-B^{-1}\|=\|A^{-1}HB^{-1}\|\le \|A^{-1}H\|\|B^{-1}\|\le \|A^{-1}\|\|H\|\|B^{-1}\|.		
	\end{flalign}
	It remains to bound $\|B^{-1}\|$. By the Woodbury identify,
	$$B^{-1}=(A+H)^{-1}=A^{-1}-A^{-1}H(I+A^{-1}H)^{-1}A^{-1}=A^{-1}-A^{-1}H(A+H)^{-1}=A^{-1}-A^{-1}HB^{-1},$$ where the second to last inequality follows since $(I+A^{-1}B)^{-1}=(A+B)^{-1}A^{}$.  
	
	By the triangle inequality, 
	$$
	\|B^{-1}\|\le \|A^{-1}\|+\|A^{-1}H B^{-1}\|\le \|A^{-1}\|+\|A^{-1}H\|\|B^{-1}\|=\|A^{-1}\|+\rho\|B^{-1}\|,
	$$implying that 
	$$
	\|B^{-1}\|\le {\|A\|^{-1}}/{(1-\rho)}.
	$$Placing the above into equation \eqref{eq:matbound} then yields the result. 
\end{proof}

Lemma \ref{lem:LIE} uses Lemmas \ref{lem:four_sim}-\ref{lem:matlem} to bound the variance of a quadratic form with a random weighting matrix. 
\begin{lemma}\label{lem:LIE}
	Let $Z\in\mathbb{R}^d$ with $\mu=\E[Z]$, and variance $\Sigma$, and let $Z_1,\dots,Z_N$ be iid observations with the same distribution as $Z$. Also, let $Z_{i,l}$ (resp., $\mu_{i,l}$), $l=1,\dots, d$, denote the coordinates of $Z_i$ (resp., $\mu$), and denote by $\overline{Z}$ their sample mean. Let $\widehat\Sigma$ be a random covariance matrix such that $\widehat\Sigma=A^\top A$ for some (random) square-root matrix $A$. For some $\sigma>0$, and each $j=1,\dots,d$, the following moment assumptions are satisfied: for each $l=1,\dots,d$, and $j\in\{2,4\}$, $\E[|Z_{i,l}-\mu_{i,l}|^j]\le \sigma^j$, $|\mu_l|\le \sigma$, $\|\Sigma\|\le \sigma^2$, and $$\E\|\widehat{\Sigma}-\Sigma\|\le C\sigma^2/N,\quad\E\|\widehat{\Sigma}-\Sigma\|^2\le C\sigma^4/N.$$ If $\|\Sigma^{}(\widehat\Sigma^{-1}-\Sigma^{-1})\|<1$ as $N\rightarrow+\infty$, then
	$$
	\E[\overline{Z}^\top \widehat\Sigma^{-1}\overline{Z}]=\mu^\top \Sigma^{-1}\mu+O(\|\sigma^2\Sigma^{-1}\|^2/N)
	$$
	and 
	$$
	\Var(\overline{Z}^\top \widehat\Sigma^{-1}\overline{Z})=O( \|\sigma^2\Sigma^{-1}\|^4/N).
	$$
\end{lemma}

\begin{remark}\normalfont The moment assumptions in Lemma \ref{lem:LIE} for $Z_{i}$ are precisely the same as those required in Lemma \ref{lem:four_sim}. The condition  $\|\Sigma^{}(\widehat\Sigma^{-1}-\Sigma^{-1})\|<1$, as $N\rightarrow+\infty$ is required in order to apply the perturbation result in Lemma \ref{lem:matlem}. Such a condition is useful as it allows us to bound the moments of  $\overline{Z}^\top \widehat\Sigma^{-1}\overline{Z}$ using  moment of $\widehat{\Sigma}-\Sigma$, rather than moments of $\widehat{\Sigma}^{-1}-\Sigma^{-1}$, which is  less interpretable. Along with the other conditions, the moment assumptions on $\widehat{\Sigma}-\Sigma$ is enough to control the first two moments of  $\overline{Z}^\top \widehat\Sigma^{-1}\overline{Z}$, which depend on the variance $\overline{Z}$, as well as the variance of $\widehat\Sigma$. 
\end{remark}

\begin{proof}
	We first prove the mean result. From the law of iterated expectations, and properties of quadratic forms
	\begin{flalign*}
		\E[\overline{Z}^\top \widehat\Sigma^{-1}\overline{Z}]&=\E\{\E[\overline{Z}^\top \widehat\Sigma^{-1}\overline{Z}\mid\widehat\Sigma^{-1}]\}\\&=\E\left[\mu^\top \widehat{\Sigma}^{-1}\mu+\tr\{ \text{Cov}(\overline{Z})\widehat\Sigma^{-1}\}\right]\\&=\mu^\top \Sigma^{-1}\mu+N^{-1}\tr (\Sigma\Sigma^{-1})+\E\left[\mu^\top (\widehat\Sigma^{-1}-\Sigma^{-1})\mu+N^{-1}\tr [\Sigma (\widehat\Sigma^{-1}-\Sigma^{-1})]\right]. 
	\end{flalign*}

For $A,B$ real, square matrices of similar dimension, and $\|A\|_F$ the Frobenius norm of $A$, recall the inequality 
$$
|\tr (AB)|\le \|A\|_F\|B\|_F
$$Over matrices of fixed dimension, all matrix norms are equivalent, and we have that $$|\tr(AB)|\le c\|A\|\|B\|$$ for some finite $c>0$ and our chosen matrix norm $\|\cdot\|$. Using the triangle inequality and the above inequality then yields
	\begin{flalign*}
		|\mu^\top [\widehat\Sigma^{-1}-\Sigma^{-1}]\mu+N^{-1}\tr\{ \Sigma (\widehat\Sigma^{-1}-\Sigma^{-1})\}|\le& |\mu^\top (\widehat\Sigma^{-1}-\Sigma^{-1})\mu|+N^{-1}|\tr\{\Sigma (\widehat\Sigma^{-1}-\Sigma^{-1})\}|\\\lesssim& \|\mu\mu^\top\|\|\widehat\Sigma^{-1}-\Sigma^{-1}\|+N^{-1}\|\Sigma\|\|\widehat\Sigma^{-1}-\Sigma^{-1}\|.
	\end{flalign*}
	By Lemma \ref{lem:matlem}, for $N$ large enough such that $\|\Sigma(\widehat\Sigma^{-1}-\Sigma^{-1})\|<1$, 
	\begin{equation}\label{eq:perturb}
		\|[\widehat\Sigma^{-1}-\Sigma^{-1}]\|\lesssim \|\Sigma^{-1}\|^2\|\widehat\Sigma-\Sigma\|. 
	\end{equation}
	Applying equation \eqref{eq:perturb} then yields 
	\begin{flalign*}
		|\mu^\top [\widehat\Sigma^{-1}-\Sigma^{-1}]\mu+N^{-1}\tr \Sigma [\widehat\Sigma^{-1}-\Sigma^{-1}]|\lesssim&\{\|\mu\mu^\top\|+N^{-1}\|\Sigma\|\}\|\widehat\Sigma^{-1}-\Sigma^{-1}\|\\\lesssim&\{\|\mu\mu^\top\|+N^{-1}\|\Sigma\|\}\|\Sigma^{-1}\|^2\|\widehat\Sigma-\Sigma\|.
	\end{flalign*}Taking expectations,  
	\begin{flalign*}
		\E|\mu^\top [\widehat\Sigma^{-1}-\Sigma^{-1}]\mu+N^{-1}\tr\{ \Sigma (\widehat\Sigma^{-1}-\Sigma^{-1})\}|\lesssim&\|\mu\mu^\top\|\|\Sigma^{-1}\|^2\sigma^2/N+\|\Sigma\|\|\Sigma^{-1}\|^2\sigma^2/N^2\\ \lesssim &
		\sigma^4\|\Sigma^{-1}\|^2/N+\sigma^4\|\Sigma^{-1}\|^2/N^2, 
	\end{flalign*}where we have used the fact that $\|\mu\mu^\top\|\lesssim\|\mu^2\|^2\lesssim \sigma^2$, and $\|\Sigma\|\lesssim \sigma^2$.

	To upper bound the variance, we use the law of iterated variance:
	\begin{flalign}\label{eq:varterm}
		\text{Var}(\overline{Z}^\top\widehat\Sigma^{-1}\overline{Z})=\E\{\text{Var}(\overline{Z}^\top\widehat\Sigma^{-1}\overline{Z}\mid\widehat\Sigma)\}+\text{Var}[\E(\overline{Z}^\top\widehat\Sigma^{-1}\overline{Z}\mid\widehat\Sigma)],
	\end{flalign}and control each term separately. For the first term in \eqref{eq:varterm}, we can directly apply Lemma \ref{lem:four_sim}, conditional on $\widehat\Sigma$ and \eqref{eq:perturb}, to obtain the bound 
	\begin{flalign*}
		\text{Var}(\overline{Z}^\top\widehat\Sigma^{-1}\overline{Z}\mid\widehat\Sigma)\lesssim \|\sigma^2\widehat\Sigma^{-1}\|^2/N&\leq \frac{\sigma^4}{N}\left\{\|(\widehat\Sigma^{-1}-\Sigma^{-1})\|^2+\|\Sigma^{-1}\|^2\right\}\\&\lesssim \frac{\sigma^4}{N}\|\Sigma^{-1}\|^2\left(1+\|\Sigma^{-1}\|^2\|\widehat{\Sigma}-\Sigma\|^2\right);
	\end{flalign*}and applying the moment bound on $\widehat\Sigma-\Sigma$ we have  
	\begin{equation*}
		\E\{\text{Var}(\overline{Z}^\top\widehat\Sigma^{-1}\overline{Z}\mid\widehat\Sigma)\}\lesssim\frac{\sigma^4}{N} \|\Sigma^{-1}\|^2(1+\|\Sigma^{-1}\|^2\sigma^4/N)=O(N^{-1}\|\sigma^2\Sigma^{-1}\|\vee N^{-2} \|\sigma^2\Sigma^{-1}\|^4).
	\end{equation*}Since $\|\sigma^2\Sigma^{-1}\|<\infty$, the first term in \eqref{eq:varterm} is upper bounded by $O(N^{-1}\|\sigma^2\Sigma^{-1}\|)$ for all $N$ large enough.
	
	To bound the second term in \eqref{eq:varterm}, use the proof of the first expectation, conditional on $\widehat\Sigma$, to deduce 
	\begin{flalign*}
		\E(\overline{Z}^\top \widehat\Sigma^{-1}\overline{Z}\mid\widehat\Sigma^{-1})&=\mu^\top \widehat\Sigma^{-1}\mu+\tr( \Sigma\widehat\Sigma^{-1})/N\\&=\mu^\top\Sigma^{-1}\mu+\mu^\top (\widehat\Sigma^{-1}-\Sigma^{-1})\mu+N^{-1}\tr [\Sigma(\widehat\Sigma^{-1}-\Sigma^{-1})]+d/N.
	\end{flalign*}Take the variance of the above and use Cauchy-Schwartz to obtain 
	\begin{flalign}
		\Var [\E(\overline{Z}^\top \widehat\Sigma^{-1}\overline{Z}\mid\widehat\Sigma^{-1})]&\le \Var[\tr (\widehat\Sigma^{-1}-\Sigma^{-1})\mu\mu^\top]+N^{-2}\Var[\tr (\widehat\Sigma^{-1}-\Sigma^{-1})\Sigma]\nonumber\\&+2\sqrt{\Var[\tr (\widehat\Sigma^{-1}-\Sigma^{-1})\mu\mu^\top]}\sqrt{N^{-2}\Var[\tr (\widehat\Sigma^{-1}-\Sigma^{-1})\Sigma]}\label{eq:varbound}
	\end{flalign}
	
	Consider the first term in equation \eqref{eq:varbound}; from the trace inequality  
	\begin{flalign*}
		|\tr (\widehat\Sigma^{-1}-\Sigma^{-1})\mu\mu^\top|\le \|\widehat\Sigma^{-1}-\Sigma^{-1}\|\tr(\mu\mu^\top)\lesssim \sigma^2 \|\widehat\Sigma^{-1}-\Sigma^{-1}\|.
	\end{flalign*}Applying equation \eqref{eq:perturb}, we have  
	$$
|\tr [(\widehat\Sigma^{-1}-\Sigma^{-1})\mu\mu^\top ]|\lesssim \sigma^2 \|\Sigma^{-1}\|^2\|\widehat\Sigma-\Sigma\|.
	$$ Hence, using the moment bound assumption,
	\begin{flalign}
		\Var\{\tr [(\widehat\Sigma^{-1}-\Sigma^{-1})\mu\mu^\top]\}\lesssim \sigma^4\|\Sigma^{-1}\|^4\E\|\widehat\Sigma-\Sigma\|^2\lesssim \|\sigma^2\Sigma^{-1}\|^4/N. \label{eq:varbound1}
	\end{flalign}	
	For the second term in \eqref{eq:varbound}, again by the trace inequality,  and \eqref{eq:perturb}
	\begin{flalign*}
		|\tr\{ \Sigma (\widehat\Sigma^{-1}-\Sigma^{-1})\}|\lesssim \|\Sigma\|\|\Sigma^{-1}-\Sigma^{-1}\|\lesssim \sigma^2\|\Sigma^{-1}-\Sigma^{-1}\|\lesssim \sigma^2\|\Sigma^{-1}\|\|\widehat\Sigma-\Sigma\|,
	\end{flalign*}since $\|\Sigma\|\leq \sigma^2$. Thus, similar to the first term in \eqref{eq:varbound}
	\begin{flalign}
		\Var[\tr \{\Sigma(\widehat\Sigma^{-1}-\Sigma^{-1})\}]\lesssim \sigma^4\|\Sigma^{-1}\|^4\E\|\widehat\Sigma-\Sigma\|^2\lesssim \|\sigma^2\Sigma^{-1}\|^4/N. \label{eq:varbound2}
	\end{flalign}

	Applying the variance bounds in \eqref{eq:varbound1}-\eqref{eq:varbound2} in \eqref{eq:varbound}, and re-arranging terms, we obtain 
	\begin{flalign*}
		\Var \{\E(\overline{Z}^\top \widehat\Sigma^{-1}\overline{Z}\mid\widehat\Sigma^{-1})\}&\lesssim  \|\sigma^2\Sigma^{-1}\|^4/N+\|\sigma^2\Sigma^{-1}\|^4/N^3+\|\sigma^2\Sigma^{-1}\|^4/N^2\\&=O(\|\sigma^2\Sigma^{-1}\|^4/N).
	\end{flalign*}
	
\end{proof}

\begin{lemma}\label{lem:tse}
	Let $Z$ be a positive, scalar-valued random variable with mean $\mu$ that satisfies $\E\left[(Z-\mu)^2\right]\le B b^2$ for some $B,b>0$. Then, for any $0\le \lambda \le 1/b$, 
	$$
	\E\left[\exp(-\lambda Z)\right]\le \exp(-\lambda\mu)+B(\lambda b)^2.
	$$
\end{lemma}

\begin{proof}
	For any $Z\ge0$, a Taylor expansion of $\exp(-\lambda Z)$  around $Z=\mu$ yields 
	\begin{flalign*} 
		\exp (-\lambda Z) & = \exp(-\lambda\mu)\sum_{n=0}^{\infty}\frac{(-1)^n\lambda^n}{n!}(Z-\mu)^n=\mathcal{P}(Z)+\mathcal{E}(Z),
	\end{flalign*}where $\mathcal{P}(Z)$ is the truncation of the series
	$$
	\mathcal{P}(Z):=\exp(-\lambda\mu)\left\{1-\lambda(Z-\mu)\right\}
	$$
	and $\mathcal{E}(Z)$ is the Lagrange remainder term, 
	$$
	\mathcal{E}(Z):=\exp(-\lambda\zeta)\frac{\lambda^2}{2!}(Z-\mu)^2
	$$with $\zeta$ a constant between $Z$ and $\mu$. 
	
	Boundedness of $\exp(-x)$ over $x\ge0$, implies that for any $\lambda\ge0$ and $Z\ge0$
	$$
	\mathcal{E}(Z)\leq\lambda^2(Z-\mu)^2.
	$$
	Consequently,
	\begin{flalign*}
		\exp (-\lambda Z) = \mathcal{P}(Z)+\mathcal{E}(Z)\le \exp(-\lambda\mu)\left\{1-\lambda(Z-\mu)\right\}+{\lambda^2}(Z-\mu)^2.
	\end{flalign*}By hypothesis, $\E\left[(Z-\mu)^2\right]\le B b^2<\infty$ for some $B,b>0$. Take expectations of both sides to obtain  
	\begin{flalign*}
		\E\left[\exp (-\lambda Z)\right] &\le \exp(-\lambda\mu)\left\{1-\lambda\E(Z-\mu)\right\}+\lambda^2\E(Z-\mu)^2\\&\le \exp(-\lambda\mu)+B(\lambda b)^2.
	\end{flalign*}
\end{proof}

The following result demonstrates that if the regularity conditions in Assumptions \ref{ass:infeasible}-\ref{ass:prior} are satisfied for $m_n(\theta)$ and $W_n(\theta)$, then the (possibly infeasible) Q-posterior $\pi^Q_n(\theta)$ is asymptotically Gaussian. Of course, if the regularity conditions are satisfied with $m_n(\theta)=\psi_n(\theta)$ and $W_n(\theta)=V_n(\theta)$, then the result also applies to the Q-posterior $\pi^\Psi_n(\theta)$ based on the loss function $q_n(\theta)$ with $\Sigma_{\star}=\Omega_\star$  (see Section \ref{sec:general} for details).

\begin{lemma}\label{lem:twog}Under Assumptions \ref{ass:tails}-\ref{ass:prior}, as $n\rightarrow+\infty$, 
	$
	\int_{\mathcal{T}_n}\|t\||\pi^Q_{n}(\vartheta)-N\{\vartheta;0,\Sigma_\star^{-1}\}|\dt \vartheta=o_p(1).
	$
\end{lemma}
\begin{proof}[Proof of Lemma \ref{lem:twog}]
	The approach used to prove this result is similar to that given in \cite{lehmann2006theory}, Theorem 8.2, Ch 6, as well as Theorem 1 in \cite{chernozhukov2003mcmc}. Our proof differs, however, since the matrix $W_n(\theta)$ is allowed to be singular away from $\theta_\star$, and since $\Theta$ is not assumed compact. {This makes the proof most similar to Lemma 1 in \cite{frazier2022bayesian}.} 
	
Throughout the remainder of this proof, let us abuse notation and define
$$
Q_n(\theta)=-\frac{1}{2}\frac{m_n(\theta)^\top}{\sqrt{n}}W_n(\theta)^{-1}\frac{m_n(\theta)}{\sqrt{n}}.
$$
For an appropriately defined remainder term $R_n(\theta)$, we have the identity
\begin{flalign}
	Q_n(\theta)-Q_n(\theta_{\star})&= -\{m_n(\theta)/\sqrt{n}\}^\top W(\theta_\star)^{-1}\mathcal{H}(\theta_\star)^\top\Sigma_\star^{-1}\Sigma_\star\sqrt{n}(\theta-\theta_\star)-\frac{n}{2}(\theta-\theta_\star)\Sigma_\star(\theta-\theta_\star)+R_n(\theta)\nonumber
	\\&=-\frac{1}{2}t^{\top}\Sigma_{\star}t+\frac{1}{2}\frac{1}{n}Z_n^{\top}\Sigma_{\star}Z_n+R_n(\theta),\label{eq:new3}
\end{flalign}where $Z_n:=\Sigma_\star^{-1}\mathcal{H}(\theta_\star)^\top W_\star^{-1}m_n(\theta_\star)$, and $t:=\sqrt{n}(\theta-\theta_\star)-Z_n/\sqrt{n}$. Define $T_n:=\theta_\star+Z_n/n$,
\begin{flalign*}
	\omega(t)&:= Q_n\left(T_n+t/\sqrt{n}\right)-Q_n(\theta_\star)-\frac{1}{2n}Z_n^{\top}\Sigma_\star Z_n,
\end{flalign*}	
and apply \eqref{eq:new3} to see that 
$$
\omega(t)=-\frac{1}{2}t^{\top}\Sigma_{\star}t+R_n(T_n+t/\sqrt{n}).$$

Recalling $\mathcal{T}_n=\{\sqrt{n}(\theta-\theta_\star)-Z_n/\sqrt{n}:\theta\in\Theta\}$, and for $M_n(\theta)=|W_n(\theta)|^{-1/2}$, %and $T_n=\theta_\star+\Sigma_\star^{-1}Z_n/\sqrt{n}$,
\begin{flalign*}
	\pi^Q_n(t)&=\frac{M_n\left(T_n+t/\sqrt{n}\right)\exp\left\{Q_n\left(T_n+t/\sqrt{n}\right)\right\}\pi\left(T_n+t/\sqrt{n}\right)}{\int_{\mathcal{T}_n} M_n\left(T_n+t/\sqrt{n}\right)\exp\left\{Q_n\left(T_n+t/\sqrt{n}\right)\right\}\pi\left(T_n+t/\sqrt{n}\right)\dt t}\\&=\frac{M_n\left(T_n+t/\sqrt{n}\right)\exp\left\{Q_n\left(T_n+t/\sqrt{n}\right)-Q_n(\theta_\star)-\frac{1}{2n}Z_n^{\top}\Sigma_\star^{-1}Z_n\right\}\pi\left(T_n+t/\sqrt{n}\right)}{\int_{\mathcal{T}_n} M_n\left(T_n+t/\sqrt{n}\right)\exp\left\{Q_n\left(T_n+t/\sqrt{n}\right)-Q_n(\theta_\star)-\frac{1}{2n}Z_n^{\top}\Sigma_\star^{-1}Z_n\right\}\pi\left(T_n+t/\sqrt{n}\right)\dt t}\\&={M_n\left(T_n+t/\sqrt{n}\right)\exp\left\{ \omega(t) \right\}\pi\left(T_n+t/\sqrt{n}\right)}/{C_n},
\end{flalign*}where
$$
C_n=\int_{\mathcal{T}_n} M_n\left(T_n+t/\sqrt{n}\right)\exp\left\{ \omega(t) \right\}\pi\left(T_n+t/\sqrt{n}\right)\dt t.
$$

The stated result follows if
\begin{flalign*}
	\int_{\mathcal{T}_n}  \|t\|^{\gamma}\left|\pi^Q_n(t)-N\{t;0,\Sigma_{\star}^{-1}\}\right|\dt t&=C_n^{-1}J_n=o_p(1), 
\end{flalign*}
where
\begin{flalign*}
	J_{n}&=\int_{\mathcal{T}_n} \|t\|^{\gamma}\bigg{|}M_n\left(T_n+\frac{t}{\sqrt{n}}\right)\exp\left\{\omega(t)\right\} \pi_{}\left(T_n+\frac{t}{\sqrt{n}}\right)-C_nN\{t;0,\Sigma_\star^{-1}\}\bigg{|} \dt t\\&=\int_{\mathcal{T}_n} \|t\|^{\gamma}\bigg{|}M_n\left(T_n+\frac{t}{\sqrt{n}}\right)\exp\left\{\omega(t)\right\} \pi_{}\left(T_n+\frac{t}{\sqrt{n}}\right)-C_n\frac{\left|\Sigma_\star\right|^{1/2 }}{(2\pi)^{{d_\theta/2}}} \exp \left\{-\frac{1}{2} t^{\top}\Sigma_\star t\right\}\bigg{|} \dt t.
\end{flalign*}
However, $$J_{n}\leq J_{1n}+J_{2n},$$ where
\begin{flalign*}
	J_{1n}&:= \int_{\mathcal{T}_n} \|t\|^{\gamma}\left|M_n\left(T_n+\frac{t}{\sqrt{n}}\right)\exp \left\{{\omega_{}(t)}{}\right\} \pi_{}\left(T_n+\frac{t}{\sqrt{n}}\right)-\pi(\theta_\star)|M(\theta_\star)|^{\frac{1}{2}}\exp \left\{-\frac{1}{2} t^{\top} \Sigma_\star^{}t\right\} \right| \dt t\\
	J_{2n}&:=\left|C_{n}\frac{\left|\Sigma_\star\right|^{1/2 }}{(2\pi)^{{d_\theta/2}}}-\pi(\theta_\star)M(\theta_\star)\right|\int_{\mathcal{T}_n}  \|t\|^{\gamma}\exp \left\{-\frac{1}{2} t^{\top} \Sigma_\star^{} t\right\} \dt t ,
\end{flalign*}for $M(\theta_\star):=|W(\theta_\star)|^{-1/2}$.

Note that if $J_{1n}=o_{p}(1)$, then 
\begin{flalign*}
	C_n&=\pi(\theta_\star)M(\theta_\star)\int_{\mathbb{R}^{d_\theta}}\exp\left\{-\frac{1}{2}t^\top \Sigma_\star t\right\}\dt t+o_p(1)=\pi(\theta_\star)M(\theta_\star)\frac{(2\pi)^{d_\theta/2}}{|\Sigma_\star|^{1/2}}+o_p(1).
\end{flalign*}Therefore, if $J_{1n}=o_p(1)$, 
\begin{flalign*}
	J_{2n}=\left|C_{n}\frac{\left|\Sigma_\star\right|^{1/2 }}{(2\pi)^{{d_\theta/2}}}-\pi(\theta_\star)M(\theta_\star)\right|\int_{\mathcal{T}_n}  \|t\|^{\gamma}\exp \left\{-\frac{1}{2} t^{\top} \Sigma_\star^{} t\right\} \dt t =o_p(1),
\end{flalign*}since $\int_{\mathbb{R}^{d_\theta}} \|t\|^{\gamma}\exp \left\{-\frac{1}{2} t^{\top} \Sigma_\star^{} t\right\}\dt t<\infty$ for any $0\le\gamma<\infty$.

Consequently, the result follows if we can prove that $J_{1n}=o_p(1)$. To demonstrate this, we split $\mathcal{T}_n$ into three regions and analyze $J_{1n}$ over each region. For some $0\le h<\infty$ and $\delta>0$, with $\delta=o(1)$, the regions are defined as follows. Region 1: $ \|t\|\leq h$; Region 2: $  h<\|t\|\leq \delta \sqrt{n}$;  Region 3: $  \|t\|\geq \delta \sqrt{n}$.

\medskip

\noindent\textbf{\textbf{Region 1}:} Over this region the result follows if
$$
\|t\|^\gamma \left|M_n\left(T_n+\frac{t}{\sqrt{n}}\right)\exp\left\{\omega(t)\right\} \pi_{}\left(T_n+\frac{t}{\sqrt{n}}\right)-\pi(\theta_\star) M(\theta_\star)\exp \left\{-\frac{1}{2} t^{\top} \Sigma_\star^{} t\right\}\right|=o_p(1).
$$
Note that, from Assumptions \ref{ass:weight} and \ref{ass:prior},
\begin{flalign*}
	\quad \sup_{\|t\|\leq h} \left\|M_n\left(T_n+{t}/{\sqrt{n}}\right)-M(\theta_\star)\right\|=o_{p}(1),\text{ and }\sup_{\|t\|\leq h}&\left|\pi\left(T_n+{t}/{\sqrt{n}}\right)-\pi(\theta_\star)\right|=o_{p}(1),
\end{flalign*}
 and $$T_n=\theta_\star+Z_n/n=\theta_\star+o_{p}(1),$$ since $Z_n/\sqrt{n}=O_p(1)$ by Assumption \ref{ass:infeasible}. Similarly, by Assumption \ref{ass:infeasible},
$$
\sup_{\|t\|\le h}\left\|T_n+t/\sqrt{n}-\theta_\star\right\|=O_p(1/\sqrt{n})
$$
so that by Lemma \ref{lem:remain}
$$
\sup_{\|t\|\le h}|R_n(T_n+t/\sqrt{n})|=o_p(1).
$$
Hence, $J_{1n}=o_{p}(1)$ from these equivalences and the dominated convergence theorem.

\bigskip

\noindent\textbf{\textbf{Region 2}:}
If $\delta=o(1)$, then $\sup_{h\le\|t\|\le \delta \sqrt{n}}\|M_n\left(T_n+t/\sqrt{n}\right)-M(\theta_\star)\|^{}=o_{p}(1)$ by Assumption \ref{ass:infeasible} and the continuous mapping theorem.
Then,  ${J}_{1n}\leq C_{1n}+C_{2n}+C_{3n}$ for $h$ large enough and $\delta=o(1)$, where
\begin{flalign*}
	C_{1n}:=&C\int_{h\leq \|t\| \leq \delta \sqrt{n}}\|t\|^{\gamma}\exp(-t^{\top}\Sigma_\star^{}t/2)\sup _{h\le\|t\| \leq \delta \sqrt{n}}\left|\exp\left\{|R_n(T_n+t/\sqrt{n})|\right\}\left\{ \pi_{}\left(T_n+{t}/{\sqrt{n}}\right)-\pi_{}\left(\theta_\star\right)\right\}\right|\dt t \\
	C_{2n}:=&C\int_{h\leq \|t\| \leq \delta \sqrt{n}}\|t\|^{\gamma} \exp(-t^{\top}\Sigma_\star^{}t/2)\exp\left\{|R_n(T_n+t/\sqrt{n})|\right\} \pi_{}\left(T_n+{t}/{\sqrt{n}}\right) \dt t \\C_{3n}:=&C\pi_{}\left(\theta_\star\right) \int_{h\leq \|t\| \leq \delta \sqrt{n}}\|t\|^{\gamma}\exp(-t^{\top}\Sigma_\star^{}t/2)\dt t .
\end{flalign*}

The first term  $C_{1n}=o_{p}(1)$ for any fixed $h$, so that $C_{1n}=o_{p}(1)$ for $h\rightarrow+\infty$, by the dominated convergence theorem. For  $C_{3n}$, we have that for any $0\le\gamma\le2$ there exists some $h'$ large enough such that for all $h>h'$, and $\|t\|\ge h$ $$\|t\|^{\gamma}\exp\left(-t^{\top}\Sigma_{\star}t/2\right)=O(1/h).$$ Hence, $C_{3n}$ can be made arbitrarily small by taking $h$ large enough and $\delta$ small enough.

The result follows if $C_{2n}=o_p(1)$. To this end, we show that, for some $C>0$, and all $h\le\|t\|\le\delta \sqrt{n}$,  with probability converging to one (wpc1),
\begin{equation}\label{eq:bound1}
	\exp(-t^{\top}\Sigma_\star t/2)\exp\left\{|R_n(T_n+t/\sqrt{n})|\right\}\pi(T_n+t/\sqrt{n})\le C\exp\left\{-t^{\top}\Sigma_\star t/4\right\}.
\end{equation}
If equation \eqref{eq:bound1} is satisfied, then $C_{2n}$ is bounded above by
\begin{flalign*}
	C_{2n}\le &C\int_{h\le\|t\|\le\delta \sqrt{n}}\|t\|^{\gamma}\exp\left(-t^{\top}\Sigma_\star t/4\right)\dt t,
\end{flalign*}which can be made arbitrarily small for some $h$ large and $\delta$ small. To demonstrate equation \eqref{eq:bound1}, first note that by continuity of $\pi(\theta)$, Assumption \ref{ass:prior}, $\pi(T_n+t/\sqrt{n})$ is bounded over $\{t:h\le \|t\|\le\delta \sqrt{n}\}$ so that it may be dropped from the analysis. 

Now, since $\|T_n-\theta_\star\|=o_p(1)$, for any $\delta>0$, $\|T_n+t/\sqrt{n}-\theta_\star\|<2\delta$ for all $\|t\|\le\delta \sqrt{n}$ and $n$ large enough. 
Therefore, by Lemma \ref{lem:remain} there exists some $\delta'>0$ and $h$ large enough so that (wpc1)
$$
\sup_{h\le\|t\|\le\delta' \sqrt{n}}\frac{|R_n(T_n+t/\sqrt{n})|}{1+\|t+\Sigma_\star^{-1}Z_n/\sqrt{n}\|^2}\le \frac{1}{4}\lambda_{\text{min}}(\Sigma_\star),
$$where $\lambda_{\text{min}}(\Sigma_\star)$ denotes the minimum eigenvalue of $\Sigma_\star$. Since $Z_n/\sqrt{n}=O_p(1)$, we have $n^{-1}Z_n^{\top}\Sigma_\star^{-1}Z_n=O_p(1)$. Thus, on the set $h\le\|t\|\le\delta' \sqrt{n}$ 
$$
|R_n(T_n+t/\sqrt{n})|\le \frac{1}{4}\lambda_{\text{min}}(\Sigma_\star^{})+\frac{1}{4}\|t+\Sigma_\star^{-1}Z_n/\sqrt{n}\|^2\lambda_{\text{min}}(\Sigma_\star^{})\le O_p(1)+\frac{1}{4}\|t\|^2\lambda_{\text{min}}(\Sigma_\star^{}),
$$ and for some $C>0$
\begin{flalign*}
	\exp\{\omega(t)\}&\leq \exp\left\{-\frac{1}{2}t^{\top}\Sigma_{\star}^{}t+|R_n(T_n+t/\sqrt{n})|\right\}\\&\le C\exp\left\{-\frac{1}{2}t^\top \Sigma_\star^{}t+\frac{1}{4}\lambda_{\text{min}}(\Sigma_\star^{})\|t\|^2\right\}\\&\leq C\exp\left(-\frac{1}{4}t^{\top}\Sigma_\star^{}t\right).
\end{flalign*}
The result follows.

\medskip

\noindent\textbf{\textbf{Region 3}:} For $\delta \sqrt{n}$ large,
$$
\int_{\|t\|\ge \delta \sqrt{n}}\|t\|^{\gamma}N\{t;0,\Sigma_{\star}^{-1}\}\dt t$$ can be made arbitrarily small and is therefore dropped from the analysis. Using the definition of $\omega(t)$, and the identity $\theta=T_n+t/\sqrt{n}$, consider
\begin{align*}
	{J}_{1n}& :=\int_{\|t\|\ge \delta \sqrt{n}}\|t\|^{\gamma}M_{n}\left(T_n+t / \sqrt{n}\right)\exp\{\omega(t) \} \pi\left(T_n+t / \sqrt{n}\right)\dt t , \\
	& =n^{(d_\theta+\gamma)/2}\int_{\|\theta-T_n\|\ge \delta }\|\theta-T_n\|^{\gamma}M_n(\theta)\exp\left\{Q_n(\theta)-Q_n(\theta_\star)-\frac{1}{2n}Z_n^\top \Sigma_\star^{-1}Z_n\right\} \pi\left(\theta\right)\dt \theta.
\end{align*}Now,
\begin{flalign*}
	{J}_{1n}& =\exp\left\{-\frac{1}{2n}Z_n^\top \Sigma_\star^{-1}Z_n\right\}n^{(d_\theta+\gamma)/2}\int_{\|\theta-T_n\|\ge \delta }\|\theta-T_n\|^{\gamma}M_n(\theta)\exp\left\{Q_n(\theta)-Q_n(\theta_\star)\right\} \pi\left(\theta\right)\dt \theta,\\&=O_p(1)n^{(d_\theta+\gamma)/2}\int_{\|\theta-T_n\|\ge \delta }\|\theta-T_n\|^{\gamma}M_n(\theta)\exp\left\{Q_n(\theta)-Q_n(\theta_\star)\right\} \pi\left(\theta\right)\dt \theta
\end{flalign*}
since $n^{-1}Z_n^\top \Sigma_\star^{-1}Z_n=O_{p}(1)$ by Assumption \ref{ass:infeasible}.

Define $Q(\theta):=-m(\theta)^{\top}W(\theta)m(\theta)/2$ and note that $Q(\theta_\star)=0$ by Assumption \ref{ass:infeasible} and $W(\theta_\star)$ is positive-definite by Assumption \ref{ass:weight}. Thus, for any $\delta>0$,
\begin{flalign*}
	\sup_{\|\theta-\theta_\star\|\ge \delta}\frac{1}{n}\left\{Q_n(\theta)-Q_n(\theta_\star)\right\}\leq& \sup_{\|\theta-\theta_\star\|\ge \delta}2|n^{-1}Q_n(\theta)-Q(\theta)|+\sup_{\|\theta-\theta_\star\|\ge \delta}\left\{Q(\theta)-Q(\theta_\star)\right\}.
\end{flalign*}
From Assumptions~\ref{ass:infeasible} and \ref{ass:weight}, the first term converges to zero in probability. From Assumption~\ref{ass:weight}, for any $\delta>0$ there exists an $\epsilon>0$ such that
$$
\sup_{\|\theta-\theta_\star\|\ge \delta}\left\{Q(\theta)-Q(\theta_\star)\right\}\le -\epsilon.
$$
Hence, 
\begin{equation}
	\label{eq:expconv}
	\lim_{n\rightarrow+\infty}P^{(n)}_0\left[\sup_{\|\theta-\theta_\star\|\geq \delta}\exp\left\{Q_n(\theta)-Q_n(\theta_\star)\right\}\leq \exp(-\epsilon n)\right]=1.
\end{equation}
Use $T_n=\theta_\star+o_p(1)$, the definition $M_n(\theta)=| W_n(\theta)^{-1}|^{1/2}$, and equation \eqref{eq:expconv} to  obtain
\begin{align*}
	{J}_{1n} & \le \{1+o_p(1)\}O_p(1)n^{(d_\theta+\gamma)/2}\int_{\|\theta-\theta_{0}\|\ge \delta }M_n(\theta)\|\theta-\theta_\star\|^{\gamma}\pi\left(\theta\right)\exp\{Q_n(\theta)-Q_n(\theta_\star)\}\dt \theta\\&\leq O_p(1) \exp\left(-\epsilon n\right)n^{(d_\theta+\gamma)/2}\int_{\|\theta-\theta_\star\|\ge \delta }M_n(\theta)\|\theta-\theta_\star\|^{\gamma}\pi\left(\theta\right)\dt \theta\\&\leq  O_p\left\{\exp\left(-\epsilon n\right)n^{(d_\theta+\gamma)/2}\right\}%\\&\le  \exp\left(-\epsilon \sqrt{n}^2\right)\sqrt{n}^{d+\gamma}\left[\int_{\Theta}|v^2_n\Delta_n(\theta)|^{-1/2}\pi\left(\theta\right)\dt \theta\right]^{1/2}\left[\int_{\Theta }\|\theta-\theta_\star\|^{\gamma}\pi\left(\theta\right)\dt \theta\right]^{1/2}\\
	%&\leq C\exp\left(-\epsilon \sqrt{n}^2\right)\sqrt{n}^{d+\gamma}\left[\int_{ \Theta}\|\theta-\theta_\star\|^{\gamma}\pi\left(\theta\right)\dt \theta\right]^{1/2}\\
	%&\leq O_p\left\{\exp\left(-\epsilon \sqrt{n}^2\right) n^{(d_\theta+\gamma)/2}\right\}
	\\&=o_p(1); 
\end{align*}where 
the second inequality follows from the moment hypothesis in Assumption~\ref{ass:prior}.
\end{proof}

The following result is used in the proof of Lemma \ref{lem:two} and is a consequence of Proposition 1 in \cite{chernozhukov2003mcmc}.
\begin{lemma}\label{lem:remain} Under Assumption \ref{ass:infeasible}, and for $R_n(\theta)$ as defined in the proof of Lemma \ref{lem:two}, for each $\epsilon>0$ there exists a sufficiently small $\delta>0$ and $h>0$ large enough, such that
	$$
	\limsup_{n\rightarrow+\infty}P^{(n)}_0\left[\sup_{h/\sqrt{n}\le \|\theta-\theta_0\|\le\delta}\frac{|R_n(\theta)|}{1+\sqrt{n}^2\|\theta-\theta_0\|^2}>\epsilon\right]<\epsilon
	$$	and
	$$
	\limsup_{n\rightarrow+\infty}P^{(n)}_0\left[\sup_{ \|\theta-\theta_0\|\le h/\sqrt{n}}{|R_n(\theta)|}>\epsilon\right]=0.
	$$	
\end{lemma}
\begin{proof}[Proof]
	The result is a specific case of Proposition~1 in \cite{chernozhukov2003mcmc}. Therefore, it is
	only necessary to verify that their sufficient conditions are satisfied in our context.
	
	Assumptions (i)-(iii) in their result follow
	directly from Assumptions \ref{ass:infeasible} and \ref{ass:weight}, and the normality of $\sqrt{n}m_n(\theta_\star)$ in 
	Assumption~\ref{ass:infeasible}. Their Assumption~(iv) is stated as follows: for any $\epsilon>0$, there is a $\delta>0$ such that
	$$
	\limsup_{n\rightarrow+\infty}P^{(n)}_0\left\{\sup_{\|\theta-\theta'\|\le\delta}\frac{\sqrt{n}\|
		\{m_n(\theta)-m_n(\theta')\}-\{\mathbb{E}\left[m_n(\theta)\right]-\mathbb{E}
		\left[m_n(\theta')\right]\}\|}{1+n\|\theta-\theta'\|}>\epsilon\right\}<\epsilon . 
	$$
	In our context, this condition is satisfied by Assumption \ref{ass:infeasible}(iv). Hence, the result follows. 
\end{proof}

\section{Proofs of Main Results}

\begin{proof}[Proof of Lemma 1]
Consider the Q-posterior under the mean parameterization $\mu=g(\eta)=\nabla_\eta A(\eta)$, with prior beliefs $\pi(\mu)\propto\exp\{-\frac{1}{2}(\mu-\mu_0)^\top W_0^{-1}(\mu-\mu_0)\}$, where $\mu_0$ and $W_0$ are known prior hyper-parameters. Writing $\bar{S}_n=n^{-1}S_n$, we see that  
$$
n^{-1}m_n(\eta)=g(\eta)-\bar{S}_n=\mu-\bar{S}_n=n^{-1}m_n(\mu).
$$Hence, writing $\bar{S}_n=n^{-1}S_n$, the Q-posterior for $\mu$ is given by 
$$
\pi_n^Q(\mu)\propto \exp\{-\frac{n}{2}\left(\mu-\bar{S}_n\right)^\top W_n^{-1}\left(\mu-\bar{S}_n\right)\}\exp\{-\frac{1}{2}(\mu-\mu_0)^\top W_0^{-1}(\mu-\mu_0)\}.
$$
Algebraic manipulations produce 
\begin{flalign*}
&\exp\left\{-\frac{1}{2}\mu^\top\left[(W_n/n)^{-1}+W_0^{-1}\right]^{}\mu-\mu^\top\left[(W_n/n)^{-1}\bar{S}_n+\mu_0\right]\right\}	
\\&=\exp\left\{-\frac{1}{2}\mu^\top\left[(W_n/n)^{-1}+W_0^{-1}\right]^{}\mu-\mu^\top\left[(W_n/n)^{-1}+W_0^{-1}\right]\left[(W_n/n)^{-1}+W_0^{-1}\right]^{-1}\left[(W_n/n)^{-1}\bar{S}_n+\mu_0\right]\right\}\\&\propto\exp\bigg{\{}-\frac{1}{2}\left(\mu-\Sigma_n^{-1}\left[(W_n/n)^{-1}\bar{S}_n+\mu_0\right]\right)^\top \Sigma_n \left(\mu-\Sigma_n^{-1}\left[(W_n/n)^{-1}\bar{S}_n+\mu_0\right]\right)\bigg{\}}\\&=\exp\left\{-\frac{1}{2}(\mu-b_n)^\top\Sigma_n(\mu-b_n)\right\}
\end{flalign*}
for
\begin{flalign*}
\Sigma_n^{-1}&=\left[(W_n/n)^{-1}+W_0^{-1}\right]^{-1}\\&=W_0\left[n^{-1}W_n+W_0\right]^{-1}W_n\frac{1}{n},
\end{flalign*}
and where the second equality follows from the Woodbury identity, and where 
\begin{flalign*}
b_n&=\Sigma_n^{-1}\left[(W_n/n)^{-1}\bar{S}_n+\mu_0\right]\\&=W_0\left[n^{-1}W_n+W_0\right]^{-1}\bar{S}_n+W_0\left[n^{-1}W_n+W_0\right]^{-1}W_n\frac{\mu_0}{n}\mu_0.
\end{flalign*}
Thus, we see that $\pi_n^Q(\mu)=N\{\mu;b_n,\Sigma_n^{-1}\}$. For a regular exponential family, the change-of-parameter from $\mu\mapsto \eta=g^{-1}(\mu)$ exists if the model is identifiable (in $\eta$). A change of variables $\mu\mapsto \eta$ then implies 
\begin{flalign*}
	\pi^Q_n(\eta)=\pi^Q_n\{g(\eta)\}|\nabla_\eta g(\eta)|=N\{g(\eta);b_n,\Sigma_n^{-1}\}|\nabla_\eta^2 A(\eta)|,
\end{flalign*}where the second equality follows since $g(\eta)=\nabla_\eta A(\eta)$. 
\end{proof}

\begin{proof}[Proof of Theorem \ref{thm:main}]
	We prove the theorem by proving the more general result
	$$
	\int_\Theta\|\varphi(\theta)\||\pi^Q_n(\theta)-\overline\pi^Q_n(\theta)|\dt\t=O(1/N).
	$$Taking $\varphi(\theta)=1$ delivers the first result, while the second result follows since 
	$$
	\left\| \int_\Theta\varphi(\theta)\{\pi^Q_n(\theta)-\overline\pi^Q_n(\theta)\}\dt\t\right\|\le \int_\Theta\|\varphi(\theta)\||\pi^Q_n(\theta)-\overline\pi^Q_n(\theta)|\dt\t.
	$$

	Recall the definitions $$\widehat Q_{n}(\theta;z):= \frac{1}{2}\frac{\widehat{m}_n(\theta;z)^\top}{\sqrt{n}}\left[n^{-1}\widehat{W}_n(\theta;z)\right]^{-1}\frac{\widehat{m}_n(\theta;z)}{\sqrt{n}}$$ and $$Q_n(\theta):= \frac{1}{2} \frac{m_n(\theta)^\top}{\sqrt{n}}\left[n^{-1}{W}_n(\theta)\right]^{-1}\frac{m_n(\theta)}{\sqrt{n}}. $$ 
	Let $\E_z$ denote expectation wrt the simulated data $z=(z_1,\dots,z_N)$ at a fixed $\theta$ and $\y$, where the dependence of $\E_z$ on $\theta$ and $\y$ is suppressed for notational simplicity. 
	We first demonstrate that, uniformly over $\Theta$,
	\begin{flalign}\label{eq:app1}
		\mathbb{E}_z[\exp\{-\widehat Q_{n}(\theta;z)\}]= \exp\left\{-Q_n(\theta)\right\}\left\{1+O(1/N)\right\}+O\left\{\sigma^2_n(\theta)/{N}\right\}.
	\end{flalign}

We first bound $\E_z[\widehat Q_{n}(\theta;z)]$ using the law of iterated expectations and Lemma \ref{lem:LIE}. In particular, taking $\overline{Z}=\widehat{m}_n(\theta;z)$, $\widehat\Sigma=\widehat{W}_N(\theta;z)$, $\sigma^2=\sigma^2_n(\theta)$, and $\mu=m_n(\theta)+O\{\sigma^2_n(\theta)/N\}$, by Lemma \ref{lem:LIE} 
	$$
	\E_z[\widehat Q_{n}(\theta;z)]= Q_n(\theta)+O\{\|\sigma^2_n(\theta) W_n(\theta)^{-1}\|^2/N\}.
	$$Since, under Assumption \ref{ass:tails}, $\sup_{\theta\in\Theta}\|\sigma^2_n(\theta) W_n(\theta)^{-1}\|<\infty$, for all $n\ge1$, it follows that $\|\sigma^2_n(\theta)W_n(\theta)^{-1}\|^2<\infty$, and we conclude that 
	\begin{equation}\label{eq:meanbound}
		\E_z[\widehat Q_{n}(\theta;z)]=\frac{1}{2}m_n(\theta)^\top W_n(\theta)^{-1} m_n(\theta)+O(1/N).	
	\end{equation}

	We now bound the variance $\E_z\{\widehat Q_{n}(\theta;z)-\E_z[\widehat Q_{n}(\theta;z)]\}^2$. Using the same definitions as above, Lemma \ref{lem:LIE} implies that, for $N$ large enough, and some $C>0$,  
	\begin{flalign*}
		\E_z\{\widehat Q_{n}(\theta;z)-\E_z[\widehat Q_{n}(\theta;z)]\}^2\le C\{\|\sigma^2_n(\theta)  W_n(\theta) ^{-1}\|^4\}\sigma^4_n(\theta)/N.
	\end{flalign*} Again, by Assumption \ref{ass:tails}, for all $n\ge1$, $\sup_{\theta\in\Theta}\|\sigma^2_n(\theta) W_n(\theta)^{-1}\|^4<\infty$, so that 
	\begin{equation}\label{eq:var}
		\E_z\{\widehat Q_{n}(\theta;z)-\E_z[\widehat Q_{n}(\theta;z)]\}^2\lesssim \sigma^4_n(\theta)/N.
	\end{equation}

	Using equations \eqref{eq:meanbound} and \eqref{eq:var}, we can upper bound $\E_z[\exp\{-\lambda\widehat{Q}_n(\theta)\}]$, for some $0\le \lambda< \sqrt{N}$, using  Lemma~\ref{lem:tse}: take $b=1/\sqrt{N}$, with $N\ge1$ by definition, and $B$ such that $\E_z \{\widehat{Q}_n(\theta)-\E_z[\widehat{Q}_n(\theta)]\}^2 \le Bb^2$; then
	\begin{flalign*}
		\mathbb{E}_z[\exp\{-\lambda\widehat{Q}_n(\theta)\}]&\le  \exp\{-\lambda \E_z[\widehat{Q}_n(\theta)]\}+{B(\lambda b)^2}\nonumber\\&=\exp[-\lambda Q_n(\theta)+O\{\lambda/N\}]+O\{\lambda^2\sigma^4_n(\theta)/N\}\nonumber\\&=\exp\{-\lambda Q_n(\theta)\}[1+O\{\lambda /N\}]+O\{\lambda^2\sigma^4_n(\theta)/N\}.
	\end{flalign*}Since the above holds for all $0\le\lambda\le1$, we take $\lambda=1$ without loss of generality for $N$ large enough, which yields equation \eqref{eq:app1}.

Define $\overline{g}_n(\theta):=\E_z\exp\{-\widehat{Q}_n(\theta;z)\}$, $ g_n(\theta):=\exp\{-Q_n(\theta)\},$ and apply equation  \eqref{eq:app1} to  obtain
	\begin{flalign}\label{eq:app3}
		|\overline{g}_n(\theta)-g_n(\theta)|\leq g_n(\theta) O\left\{1/N\right\}+O\{\sigma^4_n(\theta)/N\}.
	\end{flalign}so that
	\begin{flalign*}
		\left|\int_{\Theta} \overline{g}_n(\theta)\pi(\theta)\dt \theta-\int_{\Theta} g_n(\theta)\pi(\theta)\dt \theta\right|&\leq \int_\Theta|\overline{g}_n(\theta)-g_n(\theta)|\pi(\theta)\dt \theta\\&\lesssim \frac{1}{N}\int_{\Theta} g_n(\theta)\pi(\theta)\dt \theta+ N^{-1}\int_{\Theta} \sigma_n^4(\theta)\pi(\theta)\dt\t.
	\end{flalign*}
	
Under the hypothesis of the result, $\int_\Theta g_n(\theta)\pi(\theta)\dt \theta<\infty$ for all $n\ge1$,   so that
	\begin{flalign}
		\left|\int_{\Theta} \|\varphi(\theta)\|^{}\overline{g}_n(\theta)\pi(\theta)\dt \theta-\int_{\Theta} \|\theta\|^{}g_n(\theta)\pi(\theta)\dt \theta\right|\nonumber  &\lesssim \frac{1}{N}\int_{\Theta}\|\varphi(\theta)\|\pi_n^Q(\theta)\dt \theta+\frac{1}{N}\int_\Theta \|\varphi(\theta)\|\sigma^4_n(\theta)\pi(\theta)\dt\theta\nonumber \\&= \frac{1}{N}\int_{\Theta}\|\varphi(\theta)\|\pi_n^Q(\theta)\dt \theta+O_{}(1/N)\label{eq:neweq},
	\end{flalign}
	where the first term on the RHS of the first equation comes from dividing and multiplying the first term by $\int_\Theta g_n(\theta)\pi(\theta)\dt\theta$, which is finite for all $n\ge1$ by hypothesis; and where the equality comes about since $\int_\Theta\|\varphi(\theta)\|\sigma_n^4(\theta)\pi(\theta)\dt\theta<\infty$ by hypothesis.
	Furthermore, by hypothesis, for $n\ge1$, $\int_{\Theta}\|\varphi(\theta)\|\pi_n^Q(\theta)\dt \theta<\infty$, and
	\begin{flalign}
		\left|\int_{\Theta}\|\varphi(\theta)\|^{} \overline{g}_n(\theta)\pi(\theta)\dt \theta-\int_{\Theta}\|\varphi(\theta)\|^{} g_n(\theta)\pi(\theta)\dt \theta\right|&\lesssim \frac{1}{N}\int_{\Theta}\|\varphi(\theta)\|^{}\pi_n^Q(\theta)\dt \theta+O(1/N)\nonumber\\&= O_{}(1/N).\label{eq:result1}
	\end{flalign}
	
	It then follows from equation \eqref{eq:result1} that
	\begin{flalign}\label{eq:result2}
		\frac{\left|\int_{\Theta} \overline{g}_n(\theta)\pi(\theta)\dt \theta-\int_{\Theta} g_n(\theta)\pi(\theta)\dt \theta\right|}{\int_{\Theta} g_n(\theta)\pi(\theta)\dt\theta}= O_{}(1/N),
	\end{flalign} and so
	\begin{flalign*}
		\frac{\int_{\Theta} \overline{g}_n(\theta)\pi(\theta)\dt \theta}{\int_{\Theta} g_n(\theta)\pi(\theta)\dt\theta}=1+O_{}(1/N);\text{ }\frac{\int_{\Theta} g_n(\theta)\pi(\theta)\dt\theta}{\int_{\Theta} \overline{g}_n(\theta)\pi(\theta)\dt \theta}=1+O_{}(1/N).
	\end{flalign*}
	Write  $\overline{\pi}_n^Q(\theta)-{\pi}_n^Q(\theta)$ as 
	\begin{flalign*}
		\overline{\pi}_n^Q(\theta)-{\pi}_n^Q(\theta)=& \frac{\overline{g}_n(\theta)\pi(\theta) }{\int_{\Theta} \overline{g}_n(\theta)\pi(\theta)\dt \theta}-\frac{g_n(\theta)\pi(\theta)}{\int_{\Theta} g_n(\theta)\pi(\theta)\dt \theta} \\
		=&\left\{ \overline{g}_n(\theta)-g_n(\theta)\right\}\frac{\pi(\theta)}{\int_{\Theta} g_n(\theta)\pi(\theta)\dt \theta} \frac{\int_{\Theta} g_n(\theta)\pi(\theta)\dt \theta}{\int_{\Theta} \overline{g}_n(\theta)\pi(\theta)\dt \theta} \\
		&-g_{n}(\theta) \pi(\theta)\left(\frac{1}{\int_{\Theta} g_{n}(\theta) \pi(\theta)\dt\theta}-\frac{1}{\int_{\Theta} \overline{g}_{n}(\theta) \pi(\theta)\dt\theta}\right),
	\end{flalign*} and apply the triangle inequality to obtain
	\begin{flalign*}
		\left|\overline{\pi}_n^Q(\theta)-{\pi}_n^Q(\theta)\right| &\leq\left|\overline{g}_n(\theta)-g_n(\theta)\right| \frac{\pi(\theta)}{\int_{\Theta} \overline{g}_n(\theta)\pi(\theta)\dt\theta}\\&+\frac{\left|\int_{\Theta} \overline{g}_n(\theta)\pi(\theta)\dt\theta-\int_{\Theta} g_n(\theta)\pi(\theta)\dt\theta\right|}{\int_{\Theta} \overline{g}_n(\theta)\pi(\theta)\dt\theta} {\pi}_n^Q(\theta).
	\end{flalign*} Multiplying by $\|\varphi(\theta)\|$, integrating both sides and applying
	equations~\eqref{eq:result1} and \eqref{eq:result2}, 
	$$
	\begin{aligned}
		\int_{\Theta}\|\varphi(\theta)\|\left|\overline{\pi}_n^Q(\theta)-{\pi}_n^Q(\theta)\right| \dt \theta & \leq \frac{1}{\int_{\Theta} \overline{g}_n(\theta)\pi(\theta)\dt\theta} \int_{\Theta}\|\varphi(\theta)\|\left|\overline{g}_n(\theta)-g_n(\theta)\right|\pi(\theta) \dt\theta\\&+\frac{\left|\int_{\Theta}\overline{g}_n(\theta)\pi(\theta)\dt\theta-\int_{\Theta}g_n(\theta)\pi(\theta)\dt\theta\right|}{\int_{\Theta}\overline{g}_n(\theta)\pi(\theta)\dt\theta}\int_{\Theta}\|\varphi(\theta)\|\pi_n^Q(\theta)\dt\theta \\
		& \leq O\left(1/N\right )+\frac{\int_{\Theta}g_n(\theta)\pi(\theta)\dt\theta}{\int_{\Theta}\overline{g}_n(\theta)\pi(\theta)\dt\theta} O\left(1/N\right )\int_{\Theta}\|\varphi(\theta)\|\pi_n^Q(\theta)\dt\theta \\
		&=O_{}\left(1/N\right).
	\end{aligned}
	$$
	For $n\ge1$, $\int_{\Theta} \|\varphi(\theta)\| \pi_n^Q(\theta)<\infty$, and 
	the first term in the second inequality is $O_{}(1/N)$; 
	the second term is also $O_{}(1/N)$ because 
	${\int_{\Theta}g_n(\theta)\pi(\theta)\dt\theta}/
	{\int_{\Theta}\overline{g}_n(\theta)\pi(\theta)\dt\theta}=1+O_{}(1/N)$.
	The stated result follows.
	
\end{proof}

\begin{proof}[Proof of Lemma \ref{lem:two}]From the triangle inequality,
	\begin{flalign*}
	\int_{\Theta}\|\vartheta\||\overline\pi_n^Q(\vartheta)-N\{\vartheta;0,\Sigma_\star^{-1}\}|\dt\vartheta\le \int_{\Theta}\|\vartheta\||\overline\pi_n^Q(\vartheta)-\pi_n^Q(\vartheta)|\dt\vartheta+\int_{\Theta}\|\vartheta\||\pi_n^Q(\vartheta)-N\{\vartheta;0,\Sigma_\star^{-1}\}|\dt\vartheta.	
	\end{flalign*}The first term is $O({1}/{N})$ for $N\rightarrow$ as $n\rightarrow+\infty$ by Theorem \ref{thm:main}. The second term is $o_p(1)$ by Lemma \ref{lem:twog}.
\end{proof}

\begin{proof}[Proof of Lemma \ref{lem:three}]
	From Theorem \ref{thm:main} and the change of variables $\theta=\theta_\star+\vartheta/\sqrt{n}+Z_n/\sqrt{n},$
	\begin{flalign*}
		\overline\theta=\int_\Theta \theta\overline\pi_n^Q(\theta)\dt\theta&=\int_\Theta \theta\left\{\overline\pi_n^Q(\theta)-\pi_n^Q(\theta)\right\}\dt\theta+\int_\Theta\theta\pi_n^Q(\theta)\dt\theta\\&=O(1/N)+\int_{\mathcal{T}_n}(\theta_\star+Z_n/n+\vartheta/\sqrt{n})\pi^Q_n(\vartheta)\dt\vartheta 
	\end{flalign*}
so that $$\sqrt{n}(\overline\theta-\theta_\star)-Z_n/\sqrt{n}=O(\sqrt{n}/N)+\int_{\mathcal{T}_n} \vartheta \pi_n^Q(\vartheta)\dt \vartheta.$$ However, since $\mathcal{T}_n\rightarrow\mathbb{R}^{d_\theta}$ as $n\rightarrow\infty$, $\int_{\mathcal{T}_n}\vartheta N\{\vartheta;0,\Sigma_{\star}^{-1}\}\dt\vartheta=o(1)$, and by Lemma \ref{lem:two}
$$
\left\|\int_{\mathcal{T}_n} \vartheta\pi_n^Q(\vartheta)\dt \vartheta\right\|\le \int_{\mathcal{T}_n} \|\vartheta\||\pi_n^Q(\vartheta)-N\{\vartheta;0,\Sigma_\star^{-1}\}\dt \vartheta+o(1)=o_p(1).
$$Therefore, we have that $\|\sqrt{n}(\overline\theta-\theta_\star)-Z_n/\sqrt{n}\|=o_p(1)$. By Assumption \ref{ass:infeasible}, $Z_n/\sqrt{n}\Rightarrow N(0,\Sigma_\star^{-1})$, and the stated result follows. 
\end{proof}

\begin{proof}[Proof of Theorem \ref{lem:qtwo}]
The proof follows the same technique as that used to prove Lemma \ref{lem:twog} but with $m_n(\theta)=\psi_n(\theta)$ and $W_n(\theta)=V_n(\theta)$, and is therefore omitted. 	
\end{proof}
\end{document}